\documentclass[11pt]{article}

\usepackage[margin=1in]{geometry}
\usepackage[T1]{fontenc}
\usepackage{microtype}
\usepackage{amsmath,amssymb,amsthm,mathtools,bm}
\usepackage{booktabs,array,multirow}
\usepackage{graphicx}
\usepackage{subcaption}
\usepackage{natbib}
\usepackage{enumitem}
\usepackage{xcolor}
\usepackage{url}
\usepackage{hyperref}
\usepackage[nameinlink,capitalize]{cleveref}
\usepackage{placeins}

\hypersetup{
  colorlinks=true,
  linkcolor=blue,
  citecolor=blue,
  urlcolor=blue
}
\graphicspath{{./}}

\newcommand{\R}{\mathbb{R}}
\newcommand{\Pp}{\mathbb{P}}
\newcommand{\E}{\mathbb{E}}
\newcommand{\1}{\mathbb{I}}
\newcommand{\cS}{\mathcal{S}}
\newcommand{\cN}{\mathcal{N}}
\newcommand{\cA}{\mathcal{A}}

\newcommand{\cT}{\mathcal{T}}
\newcommand{\wh}{\widehat}
\newcommand{\wt}{\widetilde}
\newcommand{\norm}[1]{\left\lVert #1\right\rVert}
\newcommand{\abs}[1]{\left\lvert #1\right\rvert}
\newcommand{\set}[1]{\left\{#1\right\}}

\numberwithin{equation}{section}

\newtheoremstyle{journal}
  {7pt}{7pt}{\itshape}{}{\bfseries}{.}{0.5em}{}
\theoremstyle{journal}
\newtheorem{theorem}{Theorem}[section]
\newtheorem{proposition}[theorem]{Proposition}
\newtheorem{corollary}[theorem]{Corollary}

\newtheorem{assumption}{Assumption}[section]
\newtheorem{definition}[theorem]{Definition}
\newtheorem{remark}[theorem]{Remark}

\title{Cluster-Based Dimensionality Reduction by Nonparametric Distributional Screening}
\author{
Sanoja Jha$^*$,
Rishikesh Muralimohan$^*$,
Praveen Athauda Arachchi$^*$,
Abhishek Bhattacharjee
\\
\small
\texttt{sanoja.pg26188@gipe.ac.in};
\texttt{rishikesh.muralimohan@gmail.com};
\texttt{praveenathaudaarachchi@gmail.com}; \\
\small \texttt{abhishek@theabstractmath.com}
}
\date{}

\begin{document}
\maketitle

\begingroup
\renewcommand{\thefootnote}{\fnsymbol{footnote}}
\footnotetext[1]{Equal contribution.}
\endgroup

\begin{abstract}
We consider dimensionality reduction for high-dimensional observations accompanied by a supplied partition into two or more clusters. The objective is not to construct a low-rank projection, but to retain an interpretable subset of the original coordinates that preserves the distributional information distinguishing the clusters. For each coordinate, the proposed procedure compares the cluster-specific empirical distribution functions through a several-sample Kolmogorov--Smirnov separation statistic. We formalize the resulting marginal cluster support and establish simultaneous finite-sample concentration over all coordinates, explicit bounds for false inclusions and omissions, and exact support recovery when the minimum distributional separation dominates the high-dimensional stochastic error. We also quantify the dimension inflation induced by using an unadjusted testing level and give a familywise-error-controlled version. Under a conditional sufficiency condition, sure screening preserves the full-data posterior cluster probabilities, mutual information, and Bayes risk; an additional result characterizes robustness to imperfectly estimated cluster labels. The procedure is invariant to strictly increasing coordinate transformations and can retain low-variance cluster signals that principal components may discard. We further develop average dual information, a criterion combining partition agreement after transformation with structural coverage of cluster-relevant coordinates, and derive its basic properties and consistency. Simulations illustrate the theory, the interpretability of the selected coordinates, and the distinction between cluster-directed screening and variance-directed projection.
\end{abstract}


\section{Introduction}
\label{sec:introduction}

High-dimensional data are often analyzed after the observations have been assigned to scientifically meaningful groups. The groups may be externally observed classes, treatment arms, disease subtypes, or clusters produced by a separate exploratory analysis. In these settings, the purpose of dimension reduction is frequently more specific than compression: one seeks a smaller representation that retains the variables responsible for the observed separation and remains interpretable in the original measurement scale. A projection such as principal component analysis (PCA) can reduce dimension efficiently, but its directions are chosen to explain marginal variation rather than cluster separation, and each direction is generally a dense combination of the original variables. These features can be undesirable when the scientific question concerns which measurements distinguish the groups.

A complementary literature treats high-dimensional reduction as a screening problem. Sure independence screening reduces an ultrahigh-dimensional feature set by ranking marginal associations with an outcome \citep{fanlv2008}; subsequent work developed model-free measures that detect nonlinear or distributional dependence \citep{li2012}. In classification, the Kolmogorov filter compares class-conditional distribution functions and therefore responds to differences in location, scale, skewness, or tails rather than only to mean shifts \citep{maizou2013}. The fused Kolmogorov filter extends this idea to broader response types, including multiclass outcomes \citep{maizou2015}. Related work on sparse clustering instead estimates the partition and the selected variables jointly \citep{wittentibshirani2010}, whereas supervised principal components first screens variables and then constructs linear combinations \citep{bair2006}.

The procedure studied here belongs to the Kolmogorov screening family, but it is formulated around a different target: preservation of a supplied cluster partition. This distinction matters. A coordinate is declared active when at least two cluster-specific marginal distributions differ, and the reduced data matrix contains the corresponding original columns. The statistical contribution is to connect this marginal target to a precise notion of cluster integrity. We derive finite-sample simultaneous error bounds that require no independence across coordinates, identify the rate at which the ambient dimension may grow, and show exactly how a fixed per-coordinate significance level contributes approximately \(p_0\alpha\) noise variables when \(p_0\) coordinates are inactive. Under a conditional sufficiency assumption, retaining all marginally active coordinates is enough to preserve the full posterior distribution of the cluster label and hence the Bayes decision problem. We also quantify the effect of cluster-label contamination, establish transformation invariance, and give a population construction in which PCA removes all cluster information while the proposed screen remains consistent.

A second contribution is a formal treatment of the average dual information (ADI) criterion. ADI combines two conceptually distinct notions of preservation: agreement between the partitions obtained before and after a transformation, and structural coverage of the original cluster-relevant coordinates by the retained transformed representation. We make the cluster-agreement term invariant to arbitrary relabeling, define the structural term through a dependency map, and establish boundedness, dominance, nuisance-augmentation, linear-growth, and consistency properties. The criterion is intentionally reference-dependent: it measures fidelity to a specified partition and a specified set of cluster-relevant source variables, not an absolute amount of Shannon information.

The simulations in the original study are particularly informative when read through the theory. In the Gaussian screening experiment, most reported reduced dimensions are close to the oracle signal count plus the expected number of unadjusted false positives. Weak equal-variance location shifts produce the anticipated loss of power, whereas variance differences remain detectable even when the means coincide. In downstream classification, the full and reduced logistic models have nearly indistinguishable threshold-performance curves after eight of twelve variables are removed. In a \(10{,}000\)-variable nonlinear experiment, the cluster-directed screen substantially outperforms an unsupervised PCA pipeline. Finally, the ADI experiments show increasing scores as informative coordinates are restored and, in a controlled sequence, little visible sensitivity to the number of nuisance coordinates.

The remainder of the paper is organized as follows. \Cref{sec:method} defines the target and the screening procedure. \Cref{sec:theory} develops the finite-sample and asymptotic theory. \Cref{sec:adi} formalizes ADI. \Cref{sec:simulation} revisits the simulation experiments in detail. \Cref{sec:applications} discusses scientific use and implementation, and \Cref{sec:discussion} concludes. All proofs are given in \Cref{app:proofs}.

\section{Cluster-directed distributional screening}
\label{sec:method}

\subsection{Sampling framework and marginal cluster support}
\label{subsec:setup}

Let \(r\geq 2\) be the number of clusters and \(p\) the number of measured coordinates. After reordering the rows by cluster, write the observations as
\[
  \bm X_{gi}=(X_{gi1},\ldots,X_{gip})^\top\in\R^p,
  \qquad i=1,\ldots,n_g,\quad g=1,\ldots,r,
\]
where \(n=\sum_{g=1}^r n_g\) and \(n_{\min}=\min_g n_g\). The following assumption is maintained throughout the screening theory.

\begin{assumption}[Stratified sampling]
\label{ass:sampling}
For each \(g\), the vectors \(\bm X_{g1},\ldots,\bm X_{gn_g}\) are independent and identically distributed according to a distribution \(P_g\) on \(\R^p\), and the samples from distinct clusters are independent. Dependence among the \(p\) coordinates of the same observation is unrestricted.
\end{assumption}

For coordinate \(j\), let
\[
  F_{gj}(x)=P_g(X_{gij}\leq x),
  \qquad g=1,\ldots,r,
\]
be the cluster-specific marginal distribution function. Define the population separation
\begin{equation}
\label{eq:population-separation}
  \Delta_j
  =\max_{1\leq g<h\leq r}\norm{F_{gj}-F_{hj}}_\infty,
  \qquad
  \norm{H}_\infty=\sup_{x\in\R}\abs{H(x)}.
\end{equation}
The marginal cluster support is
\begin{equation}
\label{eq:support}
  \cS=\set{j\in\{1,\ldots,p\}:\Delta_j>0},
  \qquad s=\abs{\cS},
\end{equation}
and \(\cN=\cS^c\), with \(p_0=\abs{\cN}=p-s\). Thus \(j\in\cN\) if and only if all \(r\) cluster-specific marginal distributions for coordinate \(j\) are identical.

\begin{remark}[The target is marginal]
\label{rem:marginal}
The equality \(F_{1j}=\cdots=F_{rj}\) means that coordinate \(j\) has no marginal distributional association with the supplied cluster label. It does not imply that \(j\) is irrelevant in every joint sense. Interaction-only structures, such as an exclusive-or relation in which each coordinate is marginally independent of the label but the pair determines the label, are outside the target in \eqref{eq:support}. This limitation is shared by marginal screening procedures and motivates conditional or iterative extensions when such structures are scientifically plausible.
\end{remark}

For inferential interpretations of the tests, the supplied labels must not be chosen adaptively from the same coordinate values without accounting for that choice.

\begin{assumption}[Exogenous or independently constructed partition]
\label{ass:labels}
The cluster labels used by the screening procedure are externally observed, treated as fixed by design, or obtained from data independent of the coordinate values entering the featurewise tests. If a clustering algorithm is used on the same observations, an independent pilot sample or an appropriate sample-splitting scheme is used to construct and transfer the partition.
\end{assumption}

The concentration results below remain algebraically valid for fixed index sets, but \Cref{ass:labels} is needed for ordinary null calibration and scientific interpretation of the resulting \(p\)-values. When labels are estimated with error, \Cref{thm:estimated-labels} gives a deterministic robustness correction.

\subsection{The screening statistic and selected representation}
\label{subsec:screening}

Let
\[
  \wh F_{gj}(x)=\frac{1}{n_g}\sum_{i=1}^{n_g}\1\{X_{gij}\leq x\}
\]
be the empirical distribution function in cluster \(g\). We use the max-pairwise several-sample Kolmogorov--Smirnov separation
\begin{equation}
\label{eq:empirical-separation}
  \wh\Delta_j
  =\max_{1\leq g<h\leq r}\norm{\wh F_{gj}-\wh F_{hj}}_\infty.
\end{equation}
This statistic tests
\[
  H_{0j}:F_{1j}=\cdots=F_{rj}
  \quad\text{against}\quad
  H_{1j}:\text{at least two of these distributions differ}.
\]
Several-sample Kolmogorov--Smirnov tests have a long history; see \citet{kiefer1959}, \citet{conover1965}, \citet{wolfnaus1973}, and \citet{bohmhornik2012}. The max-pairwise form in \eqref{eq:empirical-separation} is convenient for transparent finite-sample analysis and has the same population null target.

For a deterministic threshold \(\tau>0\), define
\begin{equation}
\label{eq:threshold-selector}
  \wh\cS_\tau=\set{j:\wh\Delta_j>\tau}.
\end{equation}
In implementation, one may instead compute an exact or Monte Carlo permutation \(p\)-value \(P_j\) for each coordinate and select
\begin{equation}
\label{eq:pvalue-selector}
  \wh\cS_\alpha=\set{j:P_j\leq\alpha}.
\end{equation}
Under \(H_{0j}\) and \Cref{ass:labels}, the pooled observations of coordinate \(j\) are exchangeable across the fixed cluster sizes, so permutation calibration is finite-sample valid, including for discrete variables and ties. The direct threshold form \eqref{eq:threshold-selector} is used for theory, while \eqref{eq:pvalue-selector} matches the implementation in the simulations.

If \(\wh\cS=\{j_1<\cdots<j_m\}\), the reduced data matrix is
\begin{equation}
\label{eq:reduced-matrix}
  \bm Y_{\mathrm{red}}
  =\left(\bm Y_{j_1},\ldots,\bm Y_{j_m}\right)\in\R^{n\times m},
\end{equation}
where \(\bm Y_j\) denotes column \(j\) of the original \(n\times p\) matrix. The method therefore preserves the original units and variable identities. Ranking by \(\wh\Delta_j\) gives a deterministic importance order, and \Cref{subsec:bootstrap-ranking} gives a bootstrap stability version.

A fixed unadjusted \(\alpha\) controls the Type I error of each coordinate but not the number of inactive coordinates retained across a large feature set. For confirmatory screening, Bonferroni or Holm adjustment controls familywise error under arbitrary dependence. Benjamini--Hochberg adjustment may be used when false discovery rate control is the target and its dependence conditions are defensible \citep{benjaminihochberg1995}. For predictive work, \(\alpha\) may be tuned by nested cross-validation, with all screening repeated inside each training fold.

\subsection{A sufficiency condition for cluster integrity}
\label{subsec:sufficiency}

Let \((\bm X,Z)\) denote a new observation and its cluster label, with \(Z\in\{1,\ldots,r\}\). The marginal support in \eqref{eq:support} becomes sufficient for preserving the full cluster decision problem under the following condition.

\begin{assumption}[Marginal cluster sufficiency]
\label{ass:sufficiency}
The marginally active coordinates satisfy
\begin{equation}
\label{eq:sufficiency}
  Z\ \perp\!\!\!\perp\ \bm X_{\cS^c}\mid \bm X_{\cS}.
\end{equation}
Equivalently, for every \(g\),
\(
  \Pp(Z=g\mid\bm X)=\Pp(Z=g\mid\bm X_{\cS})
\)
almost surely.
\end{assumption}

\Cref{ass:sufficiency} separates two logically distinct claims. The KS screen can recover the marginal support without it. The assumption is required only to conclude that this support contains all information needed for the cluster label. It allows arbitrarily strong dependence among coordinates and permits inactive coordinates to be associated with active coordinates, provided they add no label information after conditioning on \(\bm X_{\cS}\).

\subsection{Bootstrap stability ranking}
\label{subsec:bootstrap-ranking}

To quantify selection stability, draw \(B\) bootstrap samples independently within each cluster, preserving \(n_1,\ldots,n_r\). Let \(W_{bj}=1\) if coordinate \(j\) is selected in replicate \(b\), and define
\begin{equation}
\label{eq:bootstrap-score}
  \wh\pi_j^*=\frac{1}{B}\sum_{b=1}^B W_{bj}.
\end{equation}
Ordering \(\wh\pi_j^*\) yields a stability ranking. This ranking has a different interpretation from ordering \(\wh\Delta_j\): it estimates the probability that a feature survives the complete resampling-and-screening procedure at the chosen threshold. \Cref{thm:bootstrap-ranking} supplies a finite-\(B\) uniform error bound.

\section{Theoretical properties}
\label{sec:theory}

\subsection{Uniform concentration and support recovery}
\label{subsec:concentration}

The main concentration result follows from the sharp Dvoretzky--Kiefer--Wolfowitz inequality \citep{massart1990}. It is distribution-free and does not require independence across coordinates.

\begin{theorem}[Uniform concentration of the cluster separations]
\label{thm:uniform}
Under \Cref{ass:sampling}, for every \(u>0\),
\begin{equation}
\label{eq:uniform-tail}
  \Pp\left(
    \max_{1\leq j\leq p}\abs{\wh\Delta_j-\Delta_j}>u
  \right)
  \leq 2rp\exp\left(-\frac{n_{\min}u^2}{2}\right).
\end{equation}
Consequently, for every \(\eta\in(0,1)\), with probability at least \(1-\eta\),
\begin{equation}
\label{eq:uniform-radius}
  \max_{1\leq j\leq p}\abs{\wh\Delta_j-\Delta_j}
  \leq
  u_n(\eta)
  :=\left\{\frac{2}{n_{\min}}
  \log\left(\frac{2rp}{\eta}\right)\right\}^{1/2}.
\end{equation}
\end{theorem}

The bound immediately yields separate controls for false inclusion and false omission.

\begin{theorem}[Finite-sample screening errors]
\label{thm:screening-errors}
Suppose \Cref{ass:sampling} holds, \(s\geq1\), and
\begin{equation}
\label{eq:deltamin}
  \Delta_{\min}=\min_{j\in\cS}\Delta_j>0.
\end{equation}
For any threshold \(0<\tau<\Delta_{\min}\),
\begin{align}
\Pp(\wh\cS_\tau\not\subseteq\cS)
&\leq 2rp_0\exp\left(-\frac{n_{\min}\tau^2}{2}\right),
\label{eq:false-inclusion}\\
\Pp(\cS\not\subseteq\wh\cS_\tau)
&\leq 2rs\exp\left
(-\frac{n_{\min}(\Delta_{\min}-\tau)^2}{2}\right).
\label{eq:false-omission}
\end{align}
Hence
\begin{equation}
\label{eq:exact-recovery-bound}
\Pp(\wh\cS_\tau\neq\cS)
\leq
2rp_0e^{-n_{\min}\tau^2/2}
+2rse^{-n_{\min}(\Delta_{\min}-\tau)^2/2}.
\end{equation}
In particular, setting \(\tau=u_n(\eta)\) gives exact recovery with probability at least \(1-\eta\) whenever \(\Delta_{\min}>2u_n(\eta)\).
\end{theorem}

\begin{corollary}[Exact recovery with growing dimension]
\label{cor:highdim}
Let \(p=p_n\), \(r=r_n\), and \(n_{\min}=n_{\min,n}\) vary with \(n\). Put
\(
  a_n=\{\log(r_np_n)/n_{\min,n}\}^{1/2}
\)
and suppose \(r_np_n\to\infty\). Choose constants \(C>\sqrt{2}\) and \(D>\sqrt{2}\), set \(\tau_n=Ca_n\), and assume
\begin{equation}
\label{eq:beta-min-rate}
  \Delta_{\min,n}\geq(C+D)a_n
\end{equation}
for all sufficiently large \(n\). Then
\(
  \Pp(\wh\cS_{\tau_n}=\cS)\to1.
\)
Thus a fixed positive separation permits \(p_n\) to grow nearly exponentially in \(n_{\min,n}\), provided \(\log(r_np_n)=o(n_{\min,n})\).
\end{corollary}

\Cref{thm:screening-errors} also clarifies the difference between sure screening and exact recovery. A threshold below the signal level makes the omission probability small, even if some inactive variables survive. Exact recovery additionally requires the threshold to dominate the largest null fluctuation. This distinction is important for downstream prediction, where false positives can increase estimation variance but need not remove population cluster information.

\subsection{What a fixed testing level does to the selected dimension}
\label{subsec:fixed-alpha}

Let \(P_j\) be the featurewise permutation \(p\)-value and \(\wh m_\alpha=\abs{\wh\cS_\alpha}\). For active coordinates define the marginal power \(\pi_j(\alpha)=\Pp(P_j\leq\alpha)\).

\begin{proposition}[Dimension inflation at an unadjusted level]
\label{prop:fixed-alpha}
Assume each null \(p\)-value is super-uniform, so that
\(
  \Pp(P_j\leq t)\leq t
\)
for all \(j\in\cN\) and \(t\in[0,1]\). Let
\(
  V_\alpha=\sum_{j\in\cN}\1\{P_j\leq\alpha\}
\)
be the number of inactive coordinates selected. Then
\begin{align}
  \E(V_\alpha)&\leq p_0\alpha,
  \label{eq:ev}\\
  \E(\wh m_\alpha)&\leq p_0\alpha+
  \sum_{j\in\cS}\pi_j(\alpha).
  \label{eq:em}
\end{align}
If the null \(p\)-values are exactly uniform, equality holds in \eqref{eq:ev}. If they are also mutually independent, then
\begin{equation}
\label{eq:binomial}
  V_\alpha\sim\operatorname{Binomial}(p_0,\alpha)
\end{equation}
and, for every \(t>0\),
\begin{equation}
\label{eq:null-hoeffding}
  \Pp\left(\abs{V_\alpha-p_0\alpha}\geq t\right)
  \leq 2\exp\left(-\frac{2t^2}{p_0}\right).
\end{equation}
When all active coordinates have power close to one, the selected dimension is therefore centered near \(s+p_0\alpha\).
\end{proposition}

\begin{proposition}[Bonferroni familywise error control]
\label{prop:bonferroni}
Under the super-uniform null condition in \Cref{prop:fixed-alpha}, the selector
\begin{equation}
\label{eq:bonferroni-selector}
  \wh\cS_{\mathrm{Bonf}}(q)
  =\set{j:P_j\leq q/p}
\end{equation}
obeys
\begin{equation}
\label{eq:fwer}
  \Pp\left(\wh\cS_{\mathrm{Bonf}}(q)\not\subseteq\cS\right)
  \leq q
\end{equation}
for every \(q\in(0,1)\), without any assumption on dependence among coordinates or \(p\)-values.
\end{proposition}

The fixed-level procedure remains useful as an exploratory screen or as a tuning path for prediction. \Cref{prop:fixed-alpha} shows, however, that \(\alpha\) is also an implicit dimension parameter. In a dataset with thousands of inactive coordinates, \(\alpha=0.05\) can retain hundreds of nuisance variables even when every active coordinate is detected.

\subsection{Preservation of the cluster decision problem}
\label{subsec:bayes}

For a coordinate set \(T\subseteq\{1,\ldots,p\}\), let \(\bm X_T\) denote the corresponding subvector. For a bounded loss \(L(a,Z)\in[0,L_{\max}]\), define the Bayes risk based on \(\bm X_T\) by
\begin{equation}
\label{eq:bayes-risk}
  R^*(T)=
  \E\left[
  \inf_a \E\{L(a,Z)\mid\bm X_T\}
  \right].
\end{equation}
For zero-one loss, \(R^*(T)\) is the minimum classification error achievable from \(\bm X_T\).

\begin{theorem}[Cluster integrity under sure screening]
\label{thm:cluster-integrity}
Suppose \Cref{ass:sufficiency} holds and let \(T\supseteq\cS\). Then, for every cluster \(g\),
\begin{equation}
\label{eq:posterior-preservation}
  \Pp(Z=g\mid\bm X_T)
  =\Pp(Z=g\mid\bm X_{\cS})
  =\Pp(Z=g\mid\bm X)
  \quad\text{almost surely}.
\end{equation}
Consequently,
\begin{equation}
\label{eq:risk-preservation}
  R^*(T)=R^*(\cS)=R^*(\{1,\ldots,p\})
\end{equation}
for every loss for which the conditional Bayes action exists. If \(H(Z)<\infty\), then
\begin{equation}
\label{eq:mi-preservation}
  I(Z;\bm X_T)=I(Z;\bm X_{\cS})=I(Z;\bm X).
\end{equation}
\end{theorem}

\begin{corollary}[High-probability Bayes-risk preservation]
\label{cor:risk}
Let \(\wh\cS_\tau\) be computed from a training sample independent of a future pair \((\bm X,Z)\). Under Assumptions \ref{ass:sampling} and \ref{ass:sufficiency}, for every \(0<\tau<\Delta_{\min}\),
\begin{equation}
\label{eq:expected-risk-bound}
  0\leq
  \E_{\mathrm{train}}\left[
    R^*(\wh\cS_\tau)-R^*(\{1,\ldots,p\})
  \right]
  \leq
  2L_{\max}rs
  \exp\left\{-\frac{n_{\min}(\Delta_{\min}-\tau)^2}{2}\right\}.
\end{equation}
For zero-one loss, take \(L_{\max}=1\).
\end{corollary}

The result explains why a reduced model may have nearly the same population classification performance even when the unadjusted screen retains some nuisance variables. Under \Cref{ass:sufficiency}, the essential event is \(\cS\subseteq\wh\cS\); false positives do not change the population posterior, although they can affect finite-sample fitting and computation.

\subsection{Robustness to imperfect cluster labels}
\label{subsec:label-robustness}

Let \(\wh F_{gj}\) denote the empirical distribution function based on the true cluster index sets and \(\wt F_{gj}\) the corresponding empirical distribution function based on estimated labels. Define
\begin{equation}
\label{eq:label-cdf-error}
  \rho_n=
  \max_{1\leq g\leq r}\max_{1\leq j\leq p}
  \norm{\wt F_{gj}-\wh F_{gj}}_\infty
\end{equation}
and let \(\wt\Delta_j\) be the analogue of \eqref{eq:empirical-separation} formed from \(\wt F_{gj}\).

\begin{theorem}[Screening with estimated labels]
\label{thm:estimated-labels}
Suppose \Cref{ass:sampling} holds. Assume that, for constants \(\bar\rho_n\geq0\) and \(\delta_n\in[0,1]\),
\begin{equation}
\label{eq:rho-event}
  \Pp(\rho_n\leq\bar\rho_n)\geq1-\delta_n.
\end{equation}
For every \(\eta\in(0,1)\), with probability at least \(1-\eta-\delta_n\),
\begin{equation}
\label{eq:estimated-label-uniform}
  \max_{1\leq j\leq p}
  \abs{\wt\Delta_j-\Delta_j}
  \leq w_n(\eta)
  :=u_n(\eta)+2\bar\rho_n.
\end{equation}
Let \(\wt\cS_\tau=\{j:\wt\Delta_j>\tau\}\). On the same event, \(\wt\cS_\tau\subseteq\cS\) whenever \(\tau\geq w_n(\eta)\), and \(\cS\subseteq\wt\cS_\tau\) whenever \(\Delta_{\min}>\tau+w_n(\eta)\). In particular, choosing \(\tau=w_n(\eta)\) gives exact recovery whenever \(\Delta_{\min}>2w_n(\eta)\).
\end{theorem}

\begin{corollary}[A bound in terms of cluster misassignment]
\label{cor:misassignment}
Let \(I_g\) and \(\wt I_g\) be the true and estimated index sets for cluster \(g\). Suppose that, for some \(\varepsilon_n\in[0,1)\),
\begin{equation}
\label{eq:set-misassignment}
  \abs{I_g\mathbin{\triangle}\wt I_g}\leq\varepsilon_n n_g,
  \qquad
  \min\{\abs{I_g},\abs{\wt I_g}\}
  \geq(1-\varepsilon_n)n_g
\end{equation}
for every \(g\). Then deterministically
\begin{equation}
\label{eq:rho-misassignment}
  \rho_n\leq\frac{2\varepsilon_n}{1-\varepsilon_n}.
\end{equation}
Consequently, \Cref{thm:estimated-labels} applies with
\begin{equation}
\label{eq:w-misassignment}
  w_n(\eta)=u_n(\eta)+
  \frac{4\varepsilon_n}{1-\varepsilon_n}.
\end{equation}
\end{corollary}

The bound is deliberately worst-case and makes no stochastic assumptions about which observations are mislabeled. It shows that consistency survives estimated clusters when the empirical label contamination vanishes faster than the smallest marginal separation.

\subsection{Transformation invariance and a contrast with PCA}
\label{subsec:invariance-pca}

\begin{proposition}[Invariance to monotone re-expression]
\label{prop:monotone}
For each coordinate \(j\), let \(h_j:\R\to\R\) be a continuous strictly increasing bijection and transform every observation by \(X_{gij}'=h_j(X_{gij})\). Then
\begin{equation}
\label{eq:monotone-invariance}
  \Delta_j'=\Delta_j
  \quad\text{and}\quad
  \wh\Delta_j'=\wh\Delta_j
\end{equation}
for all \(j\). Therefore the threshold selector, the ranking by \(\wh\Delta_j\), and any permutation \(p\)-value based on \(\wh\Delta_j\) are unchanged.
\end{proposition}

This rank-scale invariance is useful when measurement units differ or monotone variance-stabilizing transformations are applied. It also emphasizes the target difference from variance-based reduction.

\begin{proposition}[A low-variance cluster signal discarded by PCA]
\label{prop:pca-counterexample}
Let \(Z\in\{-1,1\}\) with \(\Pp(Z=1)=\Pp(Z=-1)=1/2\). Suppose
\begin{align}
  X_1&=\mu Z+\epsilon_1,
  &\epsilon_1&\sim N(0,\sigma_1^2),
  \label{eq:pca-signal}\\
  X_j&=\epsilon_j,
  &\epsilon_j&\sim N(0,\sigma_j^2),
  \qquad j=2,\ldots,p,
  \label{eq:pca-noise}
\end{align}
where \(\mu\neq0\), all errors are mutually independent, and they are independent of \(Z\). Assume that at least \(q\) of the noise variances satisfy
\begin{equation}
\label{eq:pca-variance-order}
  \sigma_j^2>\mu^2+\sigma_1^2.
\end{equation}
Then every population PCA subspace spanned by the first \(q\) eigenvectors can be chosen entirely within the noise coordinates and is independent of \(Z\). The Bayes classification error from that PCA representation is \(1/2\). In contrast,
\begin{equation}
\label{eq:normal-ks-signal}
  \Delta_1=2\Phi\left(\frac{\abs\mu}{\sigma_1}\right)-1>0,
  \qquad
  \Delta_j=0\quad(j\geq2),
\end{equation}
so the KS selector recovers the unique cluster-relevant coordinate under the conditions of \Cref{thm:screening-errors}.
\end{proposition}

The proposition does not assert that PCA is generally inferior. It isolates the mismatch between two objectives: PCA preserves directions of high unconditional variance, whereas the proposed screen preserves coordinates with high cluster-conditional distributional separation.

\subsection{Finite-bootstrap accuracy of the stability ranking}
\label{subsec:bootstrap-theory}

Let \(\pi_j^*=\Pp^*(W_{1j}=1)\) be the conditional bootstrap selection probability given the observed data, where \(\Pp^*\) denotes bootstrap probability.

\begin{theorem}[Uniform accuracy and rank recovery]
\label{thm:bootstrap-ranking}
Conditional on the observed data, suppose the \(B\) bootstrap replicates are generated independently. Then, for every \(t>0\),
\begin{equation}
\label{eq:bootstrap-uniform}
  \Pp^*\left(
  \max_{1\leq j\leq p}\abs{\wh\pi_j^*-\pi_j^*}>t
  \right)
  \leq 2p\exp(-2Bt^2).
\end{equation}
Let \(A\subset\{1,\ldots,p\}\) be a target top-ranked set and assume the conditional gap
\begin{equation}
\label{eq:bootstrap-gap}
  \gamma=
  \min_{j\in A,\,k\notin A}(\pi_j^*-\pi_k^*)>0.
\end{equation}
Then the ranking by \(\wh\pi_j^*\) places every member of \(A\) above every nonmember with conditional probability at least
\begin{equation}
\label{eq:bootstrap-rank-prob}
  1-2p\exp\left(-\frac{B\gamma^2}{2}\right).
\end{equation}
\end{theorem}

\section{Average dual information}
\label{sec:adi}

\subsection{Partition agreement}
\label{subsec:partition-agreement}

Let \(\mathcal{K}\) be a fixed clustering algorithm mapping an \(n\times d\) data matrix to a label vector in \(\{1,\ldots,r\}^n\). Given the original matrix \(\bm Y\) and a transformation \(f(\bm Y)\), define
\[
  \bm z=\mathcal{K}(\bm Y),
  \qquad
  \bm z_f=\mathcal{K}\{f(\bm Y)\}.
\]
The same construction applies when \(\bm z\) is an externally supplied partition, in which case only \(\bm z_f\) is produced algorithmically.

For label vectors \(\bm u,\bm v\in\{1,\ldots,r\}^n\), let \(M(\bm u,\bm v)=(m_{ab})\) be their confusion matrix, with row totals \(a_k\), column totals \(b_k\), total \(N=n\), and diagonal sum \(c=\sum_km_{kk}\). The multiclass Matthews correlation coefficient is
\begin{equation}
\label{eq:multiclass-mcc}
  \phi\{M(\bm u,\bm v)\}
  =\frac{Nc-\sum_{k=1}^r a_kb_k}
  {\left\{(N^2-\sum_{k=1}^ra_k^2)
  (N^2-\sum_{k=1}^rb_k^2)\right\}^{1/2}},
\end{equation}
provided the denominator is positive \citep{gorodkin2004,jurman2012}. Because cluster labels are arbitrary, define the aligned nonnegative agreement
\begin{equation}
\label{eq:aligned-mcc}
  \Gamma(\bm u,\bm v)
  =\max_{\pi\in\mathfrak{S}_r}
  \left[\phi\{M(\bm u,\pi\circ\bm v)\}\right]_+,
  \qquad [x]_+=\max(x,0),
\end{equation}
where \(\mathfrak{S}_r\) is the permutation group. Thus \(0\leq\Gamma\leq1\), and \(\Gamma=1\) for identical partitions up to relabeling. We assume throughout this section that both partitions have nondegenerate cluster margins so that \eqref{eq:multiclass-mcc} is defined.

\subsection{Structural coverage and the ADI definition}
\label{subsec:adi-definition}

Let the transformation have output coordinates
\(
  f(\bm Y)=(f_1(\bm Y),\ldots,f_\ell(\bm Y)).
\)
Associate each output coordinate \(f_k\) with a known dependency set
\begin{equation}
\label{eq:dependency-set}
  \cA_k(f)\subseteq\{1,\ldots,p\},
\end{equation}
containing the source coordinates on which \(f_k\) depends. For a linear transformation, \(\cA_k(f)\) is the set of nonzero loadings in component \(k\). Apply the screening procedure using the same reference partition \(\bm z\) to the original and transformed coordinates, obtaining \(\wh\cS_Y\) and \(\wh\cT_f\), respectively. The retained source ancestry is
\begin{equation}
\label{eq:retained-ancestry}
  \wh\cA_f
  =\bigcup_{k\in\wh\cT_f}\cA_k(f).
\end{equation}
When \(\abs{\wh\cS_Y}>0\), define the structural coverage
\begin{equation}
\label{eq:structural-coverage}
  \wh Q_f
  =\frac{\abs{\wh\cS_Y\cap\wh\cA_f}}
  {\abs{\wh\cS_Y}}.
\end{equation}
Using the same reference partition in both screens is important: it ensures that the structural term and the partition-agreement term refer to the same cluster target.

\begin{definition}[Dual information and average dual information]
\label{def:adi}
For \(\lambda\in[0,1]\), the empirical dual information of \(f\) is
\begin{equation}
\label{eq:dual-information}
  \wh I_\lambda(f)
  =\lambda\Gamma\{\bm z,\bm z_f\}
  +(1-\lambda)\wh Q_f.
\end{equation}
The average dual information is the equal-weight case
\begin{equation}
\label{eq:adi}
  \wh I(f)=\wh I_{1/2}(f)
  =\frac{1}{2}\left[
  \Gamma\{\bm z,\bm z_f\}+\wh Q_f
  \right].
\end{equation}
Reported percentages are \(100\wh I(f)\).
\end{definition}

The first component asks whether the transformation reproduces the reference partition. The second asks whether the selected transformed coordinates retain functional access to the original coordinates identified as cluster-relevant. The structural term is binary in its default form: any retained dependence on a source coordinate counts as coverage. For dense transformations such as PCA, a weighted extension based on squared loadings may be preferable; the binary definition is retained here because it matches the proposed criterion and the simulations.

\subsection{Properties of ADI}
\label{subsec:adi-properties}

For population or oracle statements, let \(\Gamma_f\in[0,1]\) and \(Q_f\in[0,1]\) denote the corresponding partition-agreement and structural-coverage quantities, and define
\(
  I_\lambda(f)=\lambda\Gamma_f+(1-\lambda)Q_f.
\)

\begin{theorem}[Basic ADI properties]
\label{thm:adi-properties}
For every transformation \(f\) and \(\lambda\in[0,1]\):
\begin{enumerate}[label=(\roman*),leftmargin=2.2em]
\item \(0\leq I_\lambda(f)\leq1\).
\item If \(\lambda\in(0,1)\), then \(I_\lambda(f)=1\) if and only if \(\Gamma_f=Q_f=1\), and \(I_\lambda(f)=0\) if and only if \(\Gamma_f=Q_f=0\).
\item For the identity transformation, \(I_\lambda(\mathrm{id})=1\), provided the same deterministic clustering and screening rules are used on both copies of the data.
\item If \(\Gamma_f\geq\Gamma_g\) and \(Q_f\geq Q_g\), then \(I_\lambda(f)\geq I_\lambda(g)\); the inequality is strict for \(\lambda\in(0,1)\) if at least one component inequality is strict.
\item The function \(\lambda\mapsto I_\lambda(f)\) is affine, and
\begin{equation}
\label{eq:adi-lipschitz}
  \abs{I_\lambda(f)-I_\lambda(g)}
  \leq
  \lambda\abs{\Gamma_f-\Gamma_g}
  +(1-\lambda)\abs{Q_f-Q_g}.
\end{equation}
\end{enumerate}
\end{theorem}

\begin{proposition}[Nuisance augmentation and informative-coordinate growth]
\label{prop:adi-path}
Consider two transformed representations \(f\) and \(g=(f,w)\).
\begin{enumerate}[label=(\roman*),leftmargin=2.2em]
\item If adding \(w\) leaves the resulting partition unchanged and does not alter the retained source ancestry, so that \(\Gamma_g=\Gamma_f\) and \(Q_g=Q_f\), then \(I_\lambda(g)=I_\lambda(f)\) for every \(\lambda\).
\item Let \(f_0,f_1,\ldots\) be a nested sequence and suppose the oracle marginal support has size \(s\). If the transition from \(f_t\) to \(f_{t+1}\) newly covers \(b_t\) previously uncovered active source coordinates, then
\begin{equation}
\label{eq:adi-increment}
  I_\lambda(f_{t+1})-I_\lambda(f_t)
  =\lambda\{\Gamma_{t+1}-\Gamma_t\}
  +(1-\lambda)\frac{b_t}{s}.
\end{equation}
In particular, if \(b_t=b\) and \(\Gamma_t\) is constant over \(t\), then ADI is exactly linear in the number of added informative coordinates. If \(\Gamma_t\) is nondecreasing, every such addition increases ADI by at least \((1-\lambda)b_t/s\).
\end{enumerate}
\end{proposition}

Part (i) formalizes the intended insensitivity to nuisance columns. It is a property of the combined transformation, screening rule, and clustering rule; an arbitrary clustering algorithm need not be stable to arbitrary high-dimensional noise. Part (ii) explains why a nearly linear path appears in the ADI simulation when informative coordinates are added at a constant rate and the partition-agreement component changes little.

\subsection{Stability and consistency of empirical ADI}
\label{subsec:adi-consistency}

Let \(\cS\) be the oracle source support with \(s=\abs{\cS}>0\), and let \(\cA_f\) be the oracle retained ancestry for transformation \(f\). Define
\begin{equation}
\label{eq:oracle-q}
  Q_f=\frac{\abs{\cS\cap\cA_f}}{s}.
\end{equation}
Let \(\wh\Gamma_f\) and \(\Gamma_f\) be empirical and limiting aligned partition-agreement scores. Put
\begin{equation}
\label{eq:set-errors}
  a_n=\abs{\wh\cS_Y\mathbin{\triangle}\cS},
  \qquad
  b_n=\abs{\wh\cA_f\mathbin{\triangle}\cA_f}.
\end{equation}

\begin{theorem}[Finite-sample stability and consistency of ADI]
\label{thm:adi-consistency}
If \(a_n<s\), then
\begin{equation}
\label{eq:q-bound}
  \abs{\wh Q_f-Q_f}
  \leq\frac{2a_n+b_n}{s-a_n},
\end{equation}
and, for every \(\lambda\in[0,1]\),
\begin{equation}
\label{eq:adi-bound}
  \abs{\wh I_\lambda(f)-I_\lambda(f)}
  \leq
  \lambda\abs{\wh\Gamma_f-\Gamma_f}
  +(1-\lambda)\frac{2a_n+b_n}{s-a_n}.
\end{equation}
Consequently, if
\begin{equation}
\label{eq:adi-consistency-conditions}
  \frac{a_n}{s}\xrightarrow{\Pp}0,
  \qquad
  \frac{b_n}{s}\xrightarrow{\Pp}0,
  \qquad
  \wh\Gamma_f-\Gamma_f\xrightarrow{\Pp}0,
\end{equation}
then \(\wh I_\lambda(f)\xrightarrow{\Pp}I_\lambda(f)\). In particular, exact recovery of the source support and transformed ancestry reduces ADI error to the partition-agreement error alone.
\end{theorem}

The convergence of \(\wh\Gamma_f\) follows, for example, when the normalized confusion-matrix cells converge, the limiting margins are nondegenerate, and the number of clusters is fixed. The multiclass MCC is continuous in those cell probabilities, and the maximum over the finite set of label permutations preserves convergence.

\section{Simulation study}
\label{sec:simulation}

The simulations address four distinct questions. The first examines how the selected dimension responds to signal strength and the per-coordinate testing level. The second asks whether screening can remove nuisance predictors without degrading a downstream classifier. The third contrasts the proposed cluster-directed reduction with PCA in a sparse nonlinear problem. The fourth studies the behavior of ADI as informative and nuisance coordinates are varied. The reported tables and figures are the available realizations from the original study; because Monte Carlo standard errors were not recorded, the numerical comparisons are interpreted descriptively rather than as estimates of repeated-sampling averages.

\subsection{Two-cluster Gaussian screening}
\label{subsec:sim-gaussian}

There are two clusters, each containing \(1{,}000\) independent observations in \(\R^{6000}\). The first \(s=1{,}000\) coordinates are cluster-varying. For observation \(i\) and coordinate \(j\leq1000\),
\begin{equation}
\label{eq:sim1-active}
  X_{ij}^{(1)}\sim N(j,\sigma_1^2),
  \qquad
  X_{ij}^{(2)}\sim N(j+\delta,\sigma_2^2).
\end{equation}
The remaining \(5{,}000\) coordinates are independent \(N(0,1)\) noise in both clusters. The coordinate-specific baseline \(j\) does not affect the separation between clusters; all active coordinates share the same pair of centered distributions after subtracting \(j\). The experiment varies \(\delta\), \(\sigma_1\), \(\sigma_2\), and the unadjusted featurewise level \(\alpha\).

When \(\sigma_1=\sigma_2=\sigma\), the population KS separation of every active coordinate is
\begin{equation}
\label{eq:normal-location-ks}
  \Delta=2\Phi\left(\frac{\abs\delta}{2\sigma}\right)-1.
\end{equation}
Thus the equal-variance signal is weak when \(\abs\delta/\sigma\) is small. When \(\sigma_1\neq\sigma_2\), the distributions differ even at \(\delta=0\), illustrating that the procedure detects scale changes that a mean-based screen would miss.

\begin{table}[t]
\centering
\caption{Selected dimension in the two-cluster Gaussian experiment: first collection of parameter settings. The ambient dimension is \(6000\), with \(1000\) active and \(5000\) inactive coordinates.}
\label{tab:sim1}
\small
\setlength{\tabcolsep}{7pt}
\begin{tabular}{ccccc}
\toprule
\(\delta\) & \(\sigma_1\) & \(\sigma_2\) & \(\alpha\) & Selected dimension \\
\midrule
0.80 & 2.0 & 1.5 & 0.005 & 1019 \\
0.45 & 2.0 & 1.0 & 0.050 & 1230 \\
1.00 & 1.5 & 1.0 & 0.005 & 1019 \\
0.05 & 1.0 & 2.0 & 0.005 & 1019 \\
0.55 & 1.0 & 2.0 & 0.010 & 1040 \\
0.75 & 2.0 & 2.0 & 0.010 & 1040 \\
0.65 & 2.0 & 1.5 & 0.001 & 1005 \\
0.45 & 1.0 & 1.5 & 0.005 & 1019 \\
0.35 & 2.0 & 1.0 & 0.050 & 1230 \\
0.50 & 1.0 & 2.0 & 0.050 & 1230 \\
0.55 & 1.0 & 2.0 & 0.005 & 1019 \\
0.55 & 2.0 & 1.0 & 0.001 & 1005 \\
0.85 & 1.0 & 1.5 & 0.005 & 1019 \\
1.00 & 1.0 & 2.0 & 0.005 & 1019 \\
1.00 & 1.0 & 1.5 & 0.005 & 1019 \\
0.00 & 2.0 & 1.5 & 0.001 & 567 \\
0.65 & 1.0 & 1.0 & 0.010 & 1040 \\
0.60 & 1.5 & 2.0 & 0.001 & 1005 \\
0.70 & 1.5 & 1.5 & 0.001 & 1005 \\
0.60 & 1.0 & 2.0 & 0.005 & 1019 \\
\bottomrule
\end{tabular}
\end{table}

\begin{table}[t]
\centering
\caption{Selected dimension in the two-cluster Gaussian experiment: second collection of parameter settings.}
\label{tab:sim2}
\small
\setlength{\tabcolsep}{7pt}
\begin{tabular}{ccccc}
\toprule
\(\delta\) & \(\sigma_1\) & \(\sigma_2\) & \(\alpha\) & Selected dimension \\
\midrule
0.90 & 1.5 & 1.0 & 0.010 & 1040 \\
0.90 & 1.5 & 2.0 & 0.005 & 1019 \\
0.45 & 1.5 & 2.0 & 0.050 & 1230 \\
0.25 & 1.5 & 1.5 & 0.050 & 1130 \\
0.35 & 2.0 & 1.0 & 0.001 & 1005 \\
0.30 & 2.0 & 2.0 & 0.010 & 639 \\
0.85 & 1.0 & 2.0 & 0.010 & 1040 \\
0.20 & 2.0 & 2.0 & 0.005 & 171 \\
0.30 & 2.0 & 1.5 & 0.001 & 1000 \\
0.65 & 1.0 & 1.0 & 0.050 & 1230 \\
0.80 & 2.0 & 2.0 & 0.005 & 1019 \\
1.00 & 2.0 & 1.0 & 0.050 & 1230 \\
0.70 & 1.0 & 1.5 & 0.050 & 1230 \\
0.65 & 1.0 & 1.5 & 0.005 & 1019 \\
0.95 & 2.0 & 2.0 & 0.001 & 1005 \\
0.25 & 2.0 & 1.5 & 0.005 & 1015 \\
0.15 & 2.0 & 2.0 & 0.001 & 28 \\
0.30 & 1.0 & 2.0 & 0.001 & 1005 \\
0.95 & 1.5 & 1.5 & 0.005 & 1019 \\
0.35 & 1.0 & 2.0 & 0.005 & 1019 \\
\bottomrule
\end{tabular}
\end{table}

The dominant pattern in \Cref{tab:sim1,tab:sim2} is predicted by \Cref{prop:fixed-alpha}. If all \(1000\) active coordinates are detected, the selected dimension should be centered near
\[
  1000+5000\alpha,
\]
which equals \(1005\), \(1025\), \(1050\), and \(1250\) for \(\alpha=0.001,0.005,0.01,0.05\), respectively. The recurrent observed values \(1005\), \(1019\), \(1040\), and \(1230\) are close to these benchmarks. Their slight downward deviations are compatible with ordinary binomial fluctuation in the null selections and occasional active-feature misses. The near repetition of selected dimensions across many different \((\delta,\sigma_1,\sigma_2)\) settings indicates that power is essentially one in those configurations and that \(\alpha\), rather than the signal parameters, determines the residual dimension.

The departures from this plateau occur in the weakest equal-variance cases. In \Cref{tab:sim2}, \((\delta,\sigma_1,\sigma_2)=(0.15,2,2)\) at \(\alpha=0.001\) yields only \(28\) selected coordinates. Subtracting the expected \(5000\alpha=5\) null selections suggests that roughly \(23\) of the \(1000\) active coordinates were detected in this realization. The analogous descriptive active-retention rates are approximately \(14.6\%\) for \((0.20,2,2,0.005)\), \(58.9\%\) for \((0.30,2,2,0.01)\), and \(88.0\%\) for \((0.25,1.5,1.5,0.05)\). These values rise with the standardized location separation and with a less stringent test level, as expected from \eqref{eq:normal-location-ks} and \Cref{thm:screening-errors}.

The row \((\delta,\sigma_1,\sigma_2)=(0,2,1.5)\) in \Cref{tab:sim1} is also instructive. Although there is no mean shift, the two normal distributions have different variances, so all first \(1000\) coordinates remain members of the population support. The selected dimension \(567\) at \(\alpha=0.001\) reflects moderate power against a pure scale alternative rather than selection under a global null. More generally, the variance-unequal rows in both tables often return to the high-power plateau even when \(\delta\) is modest. This is the intended advantage of comparing complete marginal distributions rather than only means.

The tables also expose a practical distinction between screening and exact variable selection. At \(\alpha=0.05\), even perfect detection retains about \(250\) inactive coordinates on average. This may be acceptable before a second-stage fit, but it is not familywise-error control. The Bonferroni selector in \Cref{prop:bonferroni}, or a DKW threshold of the form in \Cref{thm:uniform}, is appropriate when the scientific goal is a sparse confirmatory set rather than a predictive screening path.

\FloatBarrier
\subsection{Twelve-variable logistic classification}
\label{subsec:sim-logistic}

The second experiment assesses whether the screen can remove variables without changing downstream classification. Twelve predictors are generated from a heterogeneous collection of distributions. Specifically, \(X_1,X_6\sim N(0,1)\), \(X_2\sim\operatorname{Gamma}(2,1)\), \(X_3,X_5\sim\operatorname{Exp}(1)\), \(X_4\sim N(4,1)\), and \(X_{10}\sim\operatorname{Unif}(0,1)\). The remaining variables follow mixtures of exponential, Student \(t\), Laplace, and beta distributions, together with \(X_{11}\sim\operatorname{Unif}(5,8)\) and a Cauchy \(X_{12}\). The mixture weights are those in the original simulation specification: \(X_7\) uses equal exponential and \(t_4\) components; \(X_8\) uses weights \(0.2,0.3,0.5\) on exponential, \(t_4\), and a second exponential component; and \(X_9\) uses equal Laplace and beta components. The notation follows the software parameterizations used in the experiment.

From \(10{,}000\) observations, define
\begin{equation}
\label{eq:logistic-latent}
  Z=X_1^2+5X_4+X_{10}+7X_5,
  \qquad
  Y=\1\{Z>\bar Z\}.
\end{equation}
The response depends only on \(X_1,X_4,X_5,X_{10}\). Screening the predictors by the two response groups at level \(0.05\) selects exactly these four variables. Logistic regression is then fitted once with all twelve predictors and once with the selected four.

\begin{figure}[t]
\centering
\begin{subfigure}[t]{0.48\textwidth}
  \centering
  \includegraphics[width=\textwidth]{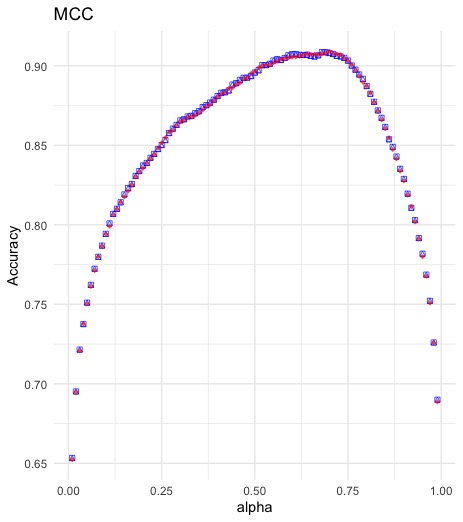}
  \caption{Matthews correlation coefficient.}
\end{subfigure}\hfill
\begin{subfigure}[t]{0.48\textwidth}
  \centering
  \includegraphics[width=\textwidth]{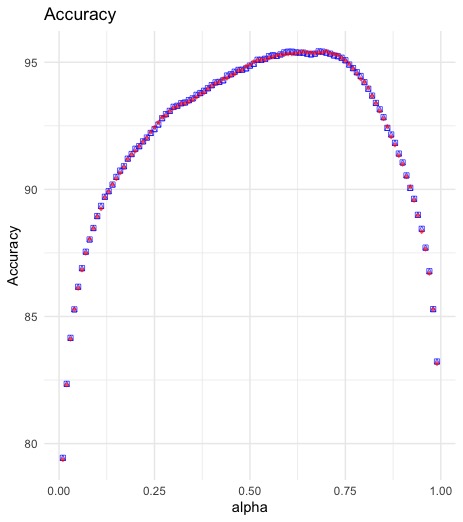}
  \caption{Classification accuracy (percent).}
\end{subfigure}
\caption{Threshold-performance profiles for the full twelve-variable logistic model and the four-variable screened model. The horizontal axis is the probability cutoff \(c\) used to convert fitted probabilities into class labels; it is distinct from the feature-screening level \(\alpha\). The two marker sequences correspond to the full and reduced models and are nearly superimposed. In the left source panel, the vertical-axis label should be read as MCC.}
\label{fig:mcc-accuracy}
\end{figure}

\Cref{fig:mcc-accuracy} varies the classification cutoff \(c\) from near zero to one. Both accuracy and MCC increase rapidly from extreme low cutoffs, remain high over a broad middle range, and then decline as the classifier predicts almost all observations in the opposite class. The two model curves are visually indistinguishable across nearly the entire path. The peak MCC is about \(0.91\), and the peak accuracy is about \(95.5\%\), attained for cutoffs in the neighborhood of \(0.6\)--\(0.7\). The optimum need not occur at \(0.5\) because the fitted logistic probabilities may be imperfectly calibrated and the data-generating rule in \eqref{eq:logistic-latent} is nonlinear.

The reduction from twelve variables to four therefore removes two thirds of the coordinates without a visible loss in this realization. The result is stronger than a coefficient-based selection claim: the screen identifies \(X_1\) because its conditional distribution differs between response groups, even though the generating mechanism uses \(X_1^2\). A linear logistic model in raw \(X_1\) is misspecified for that component, but the model-free screen still detects distributional relevance. In practice, the selected variables can be followed by appropriate basis expansion or a nonlinear classifier.

MCC is useful here because it incorporates all four cells of the binary confusion matrix and remains informative when the class proportions or error costs are uneven \citep{matthews1975}. The near overlap in both MCC and accuracy shows that the result is not an artifact of a single performance measure. For an unbiased comparison, the feature-screening level and probability cutoff should be selected within the training data, with the final test set used only once; the displayed curves are best interpreted as performance profiles.

\FloatBarrier
\subsection{Sparse nonlinear classification with \texorpdfstring{\(10{,}000\)}{10,000} predictors}
\label{subsec:sim-pca}

The third experiment places the method in an ultrahigh-dimensional, heterogeneous setting. Each of \(10{,}000\) columns contains \(1{,}500\) observations. A column family is drawn from a normal, gamma, exponential, Laplace, or Student \(t\) distribution. For the normal and gamma families, the parameters are independently drawn from \(\operatorname{Unif}(0,1)\); the exponential rate is drawn from the same distribution; and the Student degrees of freedom are drawn uniformly from \(\{1,\ldots,10\}\). The response-generating score is
\begin{equation}
\label{eq:highdim-latent}
  Z=X_1^2+4X_{50}+X_{90}^3+5X_{96}+7X_{1000}+4X_{5000},
  \qquad
  Y=\1\{Z>\operatorname{median}(Z)\}.
\end{equation}
Only six of the \(10{,}000\) predictors enter the response, and two enter through nonlinear transformations. The median split makes the two response groups approximately balanced.

A training sample of \(1{,}000\) observations and a test sample of \(500\) observations are used. Following the preceding experiment, logistic regression is fitted after dimension reduction. In the proposed pipeline, the featurewise screening level is varied over a fine grid up to \(0.01\). In the PCA pipeline, the number of retained principal components is varied up to \(500\). For each fitted model, the probability cutoff is varied and the best test-set MCC along that cutoff path is recorded. The resulting maxima are therefore descriptive profiles over the tuning paths, not performance estimates for a tuning rule chosen independently of the test set.

\begin{figure}[t]
\centering
\begin{subfigure}[t]{0.48\textwidth}
  \centering
  \includegraphics[width=\textwidth]{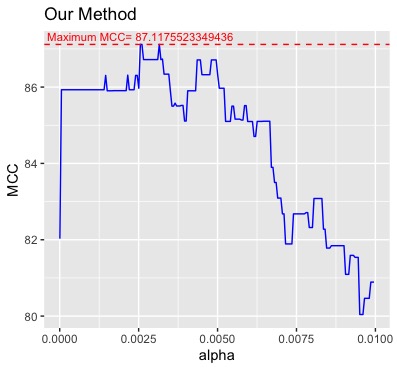}
  \caption{Cluster-directed screening.}
\end{subfigure}\hfill
\begin{subfigure}[t]{0.48\textwidth}
  \centering
  \includegraphics[width=\textwidth]{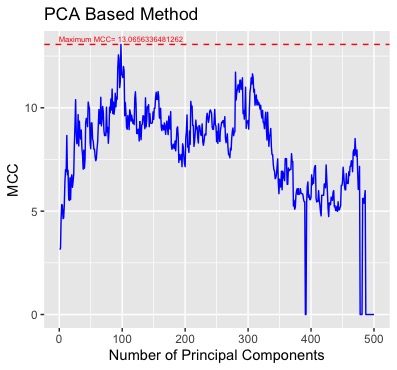}
  \caption{Principal component reduction.}
\end{subfigure}
\caption{Comparison of the proposed screening path and the PCA path in the \(10{,}000\)-predictor experiment. The vertical axes display \(100\times\mathrm{MCC}\). The dashed lines mark the largest observed values: approximately \(87.12\) for the screening pipeline and \(13.07\) for the PCA pipeline.}
\label{fig:pca-comparison}
\end{figure}

The proposed procedure reaches a maximum of approximately \(87.12\) on the \(100\times\mathrm{MCC}\) scale at a screening level near \(0.003\). Performance is lower at the smallest levels, where active variables can be omitted, and declines after the peak as increasingly many nuisance variables enter the logistic fit. This nonmonotone path is consistent with the two error terms in \Cref{thm:screening-errors}: relaxing the threshold initially reduces false omissions but eventually increases false inclusions and downstream estimation noise.

The PCA path is much weaker and considerably more irregular. Its maximum is approximately \(13.07\), near one hundred components, and several large-component fits collapse toward zero MCC. The gap is plausible under the design. PCA is unsupervised and allocates components according to unconditional covariance, whereas the response in \eqref{eq:highdim-latent} depends on six sparse coordinates drawn from distributions with heterogeneous scales. High-variance nuisance directions can dominate the principal components, while a low-variance predictor can still have a pronounced response-conditional distributional difference. \Cref{prop:pca-counterexample} gives an exact population version of this mechanism. The square and cube terms also favor a distributional screen over procedures restricted to linear mean association.

The experiment should not be interpreted as a comprehensive comparison with supervised dimension reduction. Supervised principal components, the Kolmogorov and fused Kolmogorov filters, distance-correlation screening, and sparse classifiers are closer competitors than ordinary PCA. The present result isolates the contrast intended by the paper: variance preservation and cluster preservation need not agree, especially under sparse nonlinear signals.

\FloatBarrier
\subsection{Average dual information}
\label{subsec:sim-adi}

The ADI experiment uses a full reference representation containing \(1000\) cluster-varying columns and \(6000\) random-noise columns. Reduced representations retain different numbers of the two types, and ADI is reported relative to the full representation.

\begin{table}[t]
\centering
\caption{Average dual information for representations containing different numbers of cluster-varying and noise coordinates. Values are percentages.}
\label{tab:adi}
\small
\begin{tabular}{rrr}
\toprule
Cluster-varying columns & Random-noise columns & ADI (\%) \\
\midrule
1000 (full) & 6000 (full) & 100.00000 \\
1000 & 0    & 90.65041 \\
0    & 6000 & 12.10472 \\
0    & 3000 & 7.851681 \\
1    & 3000 & 8.051276 \\
10   & 3000 & 56.66667 \\
100  & 3000 & 60.32520 \\
500  & 3000 & 76.58537 \\
1000 & 3000 & 96.91057 \\
1000 & 1000 & 92.80488 \\
\bottomrule
\end{tabular}
\end{table}

The identity representation has ADI \(100\%\) by construction. A representation containing all \(1000\) cluster-varying columns and no noise retains most of the reference information, with ADI \(90.65\%\). Noise-only representations have low scores: \(12.10\%\) with all \(6000\) noise columns and \(7.85\%\) with \(3000\). Adding a single informative coordinate to the latter changes the score only slightly, whereas adding \(10\), \(100\), \(500\), and \(1000\) informative coordinates produces progressively larger values. The sharp increase between one and ten coordinates indicates that the partition-agreement component can change nonlinearly when a small number of decisive variables becomes available; ADI is not constrained to be linear outside the controlled conditions of \Cref{prop:adi-path}.

The rows containing all \(1000\) informative coordinates also clarify the reference-dependent nature of ADI. Scores are \(90.65\%\), \(92.80\%\), \(96.91\%\), and \(100\%\) as the number of retained noise columns moves from \(0\) to \(1000\), \(3000\), and the full \(6000\). Because the reference partition was itself obtained from the full representation, reintroducing some of the same nuisance variation can make the transformed clustering more similar to that reference, even though the added columns are not cluster-varying in population. ADI therefore measures fidelity to the chosen full-data partition, not causal or intrinsic information about the latent classes. This is a feature of the definition that should be made explicit when interpreting the score.

\begin{figure}[p]
\centering
\includegraphics[page=1,width=0.30\textwidth]{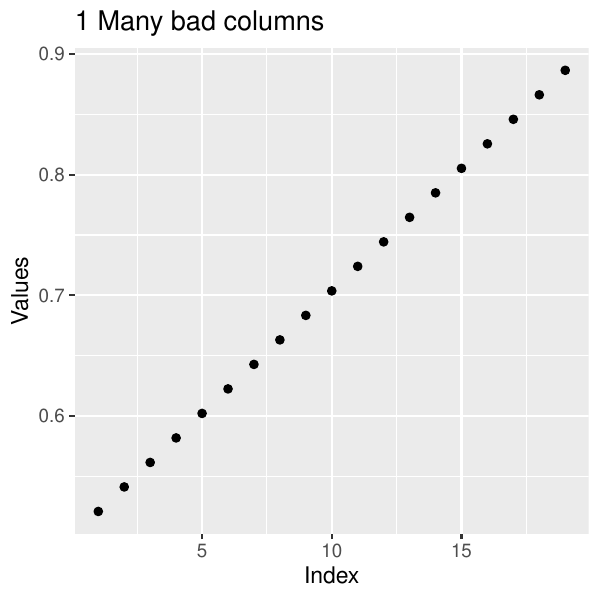}\hfill
\includegraphics[page=2,width=0.30\textwidth]{plots.pdf}\hfill
\includegraphics[page=3,width=0.30\textwidth]{plots.pdf}

\medskip
\includegraphics[page=4,width=0.30\textwidth]{plots.pdf}\hfill
\includegraphics[page=5,width=0.30\textwidth]{plots.pdf}\hfill
\includegraphics[page=6,width=0.30\textwidth]{plots.pdf}
\caption{ADI paths as informative coordinates are added while the number of nuisance columns ranges from 1 to 2501. The nearly identical, approximately linear paths indicate that nuisance-column count has little visible effect in this controlled sequence.}
\label{fig:adi-paths-a}
\end{figure}

\begin{figure}[p]
\centering
\includegraphics[page=7,width=0.30\textwidth]{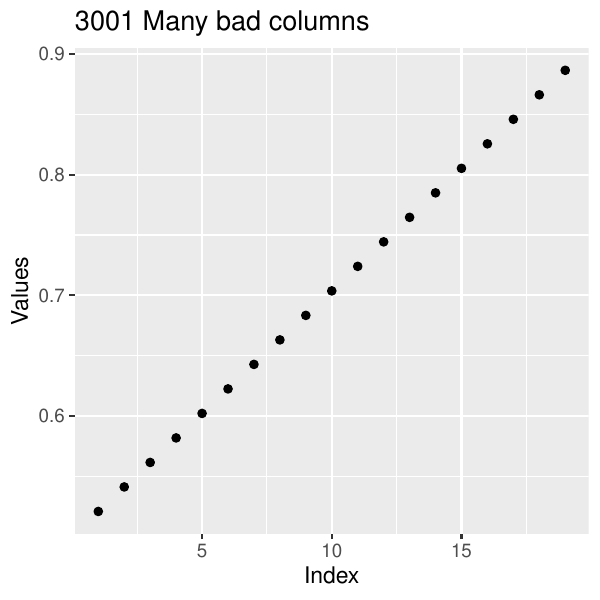}\hfill
\includegraphics[page=8,width=0.30\textwidth]{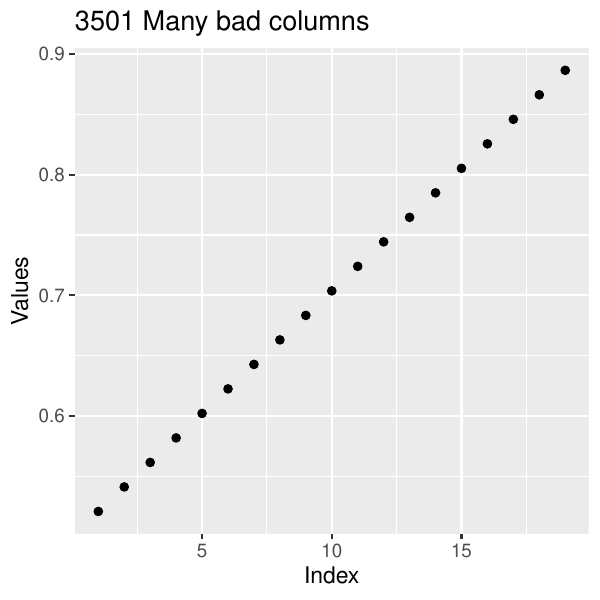}\hfill
\includegraphics[page=9,width=0.30\textwidth]{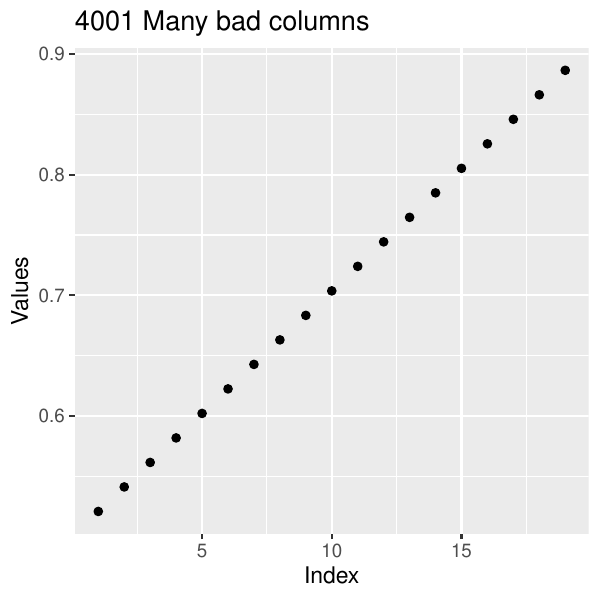}

\medskip
\includegraphics[page=10,width=0.30\textwidth]{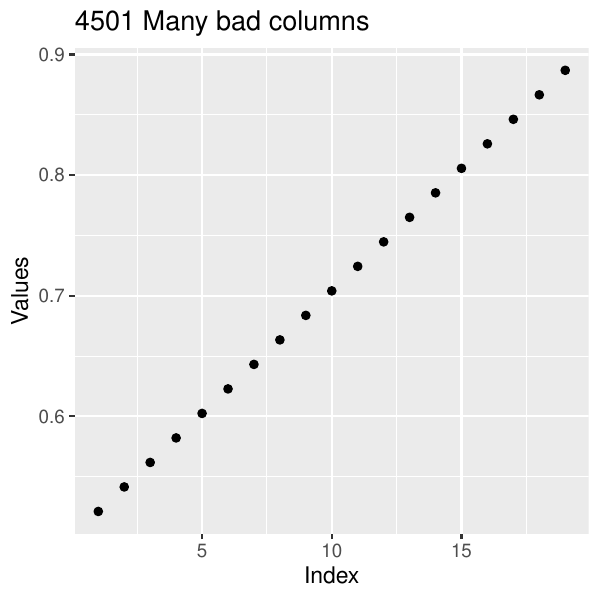}
\caption{Continuation of \Cref{fig:adi-paths-a}, with 3001, 3501, 4001, and 4501 nuisance columns. The paths remain visually unchanged.}
\label{fig:adi-paths-b}
\end{figure}

Across the ten panels in \Cref{fig:adi-paths-a,fig:adi-paths-b}, the plotted values rise almost linearly from roughly \(0.52\) to \(0.89\) as the index increases. More importantly, the paths are visually the same when the nuisance count ranges from \(1\) to \(4501\). This behavior is consistent with \Cref{prop:adi-path}: if nuisance coordinates are not selected and do not alter the induced partition, the ADI components are unchanged; if the sequence adds the same number of previously uncovered informative coordinates at each step while partition agreement is approximately stable, the structural term produces a linear increase. The fact that the paths stop below one indicates that the sequence does not reproduce both perfect partition agreement and complete structural coverage.

\FloatBarrier
\section{Scientific use and implementation}
\label{sec:applications}

The most direct applications arise when the grouping variable is scientifically meaningful and the measurements are numerous. In a case-control study, for example, the method can screen biomarkers whose marginal distributions differ between cases and controls. In a multicenter clinical study, it can identify measurements that vary across treatment arms or phenotypic subgroups. In an unsupervised workflow, a stable cluster solution obtained from an independent discovery sample can be transferred to a validation sample, and the screen can then identify a parsimonious set of variables that reproduces the cluster distinctions.

The interpretation is distributional rather than causal. Selection means that at least two group-specific marginal distributions differ. It does not establish that a variable causes the group assignment, that it contributes conditionally after other selected variables are known, or that it will remain selected under a different partition. Correlated active variables may all survive, and interaction-only variables may be missed. When conditional relevance is the goal, the screen is best viewed as a first stage followed by a joint model, conditional independence analysis, or an iterative residual-screening step.

Several implementation choices follow directly from the theory. First, the testing level should reflect the purpose. A fixed \(\alpha\) creates a useful prediction path but retains about \(p_0\alpha\) inactive variables in high-power settings. Bonferroni or Holm adjustment is appropriate for familywise control, whereas a false-discovery-rate procedure can be used for exploratory discovery under its stated dependence conditions. Second, labels obtained from the same data should be separated from feature testing by an independent sample or sample splitting; otherwise ordinary permutation calibration is optimistic. Third, the feature-selection step must be repeated inside each resampling fold when evaluating a downstream model. Fourth, the selected variables should be inspected for redundancy and scientific plausibility rather than interpreted as a minimal conditional Markov blanket.

Computationally, each coordinate requires empirical-distribution comparisons across clusters. Sorting each column gives a direct implementation with complexity on the order of \(O(pn\log n)\), and the computation is embarrassingly parallel across coordinates. When \(p\) is too large to hold the full matrix in memory, columns can be streamed because neither the statistic nor the thresholding step requires simultaneous access to all coordinates. Permutation calibration can likewise be parallelized, although the analytic threshold in \Cref{thm:uniform} avoids resampling when conservative simultaneous control is acceptable.

\section{Discussion and conclusion}
\label{sec:discussion}

Cluster-directed dimensionality reduction is most naturally viewed as a supervised screening problem with the supplied partition playing the role of a categorical response. The proposed procedure retains original coordinates whose cluster-specific empirical distribution functions differ. This target provides interpretability, responds to changes beyond the mean, and is invariant to monotone re-expression. The theory gives a finite-sample account of its high-dimensional behavior: the maximal estimation error is of order \(\{\log(rp)/n_{\min}\}^{1/2}\), exact support recovery follows from a corresponding minimum-separation condition, and a fixed unadjusted level contributes approximately \(p_0\alpha\) nuisance variables in high-power regimes.

The connection between screening and cluster integrity requires an explicit assumption. Marginal distributional activity alone does not rule out interaction-only information or conditional redundancy. Under marginal cluster sufficiency, however, any selected set containing the active support preserves the complete posterior distribution of the label, the mutual information with the label, and the Bayes risk. This result formalizes why false positives are less damaging than false negatives for population cluster preservation, while still acknowledging that false positives can degrade finite-sample model fitting. The label-contamination result further shows how imperfect preliminary clustering enters the signal requirement as an additive perturbation.

Average dual information complements the screening rule by evaluating arbitrary transformations relative to a reference partition. Its agreement component measures whether the partition is reproduced; its structural component measures whether the retained transformed coordinates continue to depend on the source variables identified as cluster-relevant. The criterion is bounded, monotone under componentwise dominance, stable under nuisance augmentation when the clustering and screening rules are themselves stable, and consistent when the relevant supports and confusion proportions are consistently estimated. Its reference dependence should be retained in interpretation. A high ADI means fidelity to a chosen representation and partition, not an absolute information-theoretic optimum.

The simulations support these conclusions but also delimit the evidence. The selected dimensions in the Gaussian study align closely with the false-positive accounting implied by the testing level, and weak equal-variance signals behave as the KS separation predicts. The logistic experiment shows that substantial reduction can leave downstream performance unchanged, while the high-dimensional comparison illustrates the danger of substituting variance preservation for label preservation. The ADI paths behave as expected under controlled nuisance augmentation. These are representative runs without Monte Carlo uncertainty, and the PCA comparison is not a substitute for benchmarking against supervised principal components, fused Kolmogorov screening, distance-correlation screening, sparse discriminant analysis, and modern nonlinear learners.

Several extensions are immediate. Conditional or iterative distributional screening can address variables with weak marginal but strong joint effects. Cross-fitted clustering can make the method inferentially valid when the partition is estimated. A weighted ADI structural term can distinguish weak from strong dependence in dense transformations, particularly for principal components. Mixed discrete-continuous data can be handled naturally through permutation calibration. Finally, repeated simulation experiments and real-data analyses are needed to assess tuning stability, computational cost, and scientific reproducibility across realistic dependence structures. Within its stated marginal target, the method provides a transparent route from a high-dimensional partition to an interpretable, cluster-preserving coordinate representation.

\appendix
\section{Proofs}
\label{app:proofs}

\subsection{Proof of the uniform concentration result}

\textbf{Proof of Theorem \ref{thm:uniform}.}
For each cluster \(g\) and coordinate \(j\), define the empirical-process error
\begin{equation}
\label{eq:proof-Egj}
  E_{gj}=\norm{\wh F_{gj}-F_{gj}}_\infty.
\end{equation}
Under \Cref{ass:sampling}, the \(n_g\) observations entering \(\wh F_{gj}\) are independent and identically distributed with distribution function \(F_{gj}\). The sharp Dvoretzky--Kiefer--Wolfowitz inequality therefore gives, for every \(\varepsilon>0\),
\begin{equation}
\label{eq:proof-dkw-one}
  \Pp(E_{gj}>\varepsilon)
  \leq 2\exp(-2n_g\varepsilon^2)
  \leq 2\exp(-2n_{\min}\varepsilon^2).
\end{equation}
No relation among different coordinates is used in this inequality.

We next compare the empirical and population pairwise separations. Fix \(j\) and a pair \(g<h\). The reverse triangle inequality in the normed space of bounded functions gives
\begin{align}
&\left|
  \norm{\wh F_{gj}-\wh F_{hj}}_\infty
  -\norm{F_{gj}-F_{hj}}_\infty
\right|
\nonumber\\
&\qquad\leq
\norm{(\wh F_{gj}-\wh F_{hj})-(F_{gj}-F_{hj})}_\infty
\nonumber\\
&\qquad=
\norm{(\wh F_{gj}-F_{gj})-(\wh F_{hj}-F_{hj})}_\infty
\nonumber\\
&\qquad\leq
\norm{\wh F_{gj}-F_{gj}}_\infty
+\norm{\wh F_{hj}-F_{hj}}_\infty
\nonumber\\
&\qquad=E_{gj}+E_{hj}.
\label{eq:proof-pair-bound}
\end{align}
For two finite collections of real numbers \(\{a_k\}\) and \(\{b_k\}\),
\(
  |\max_k a_k-\max_k b_k|\leq\max_k|a_k-b_k|.
\)
Applying this fact to the pairs \((g,h)\) and then using \eqref{eq:proof-pair-bound} yields
\begin{equation}
\label{eq:proof-delta-bound}
  \abs{\wh\Delta_j-\Delta_j}
  \leq
  \max_{g<h}(E_{gj}+E_{hj})
  \leq 2\max_{1\leq g\leq r}E_{gj}.
\end{equation}
Consequently,
\begin{align}
&\Pp\left(
  \max_{1\leq j\leq p}\abs{\wh\Delta_j-\Delta_j}>u
\right)
\nonumber\\
&\qquad\leq
\Pp\left(
  \max_{1\leq g\leq r}\max_{1\leq j\leq p}E_{gj}>\frac{u}{2}
\right)
\nonumber\\
&\qquad\leq
\sum_{g=1}^r\sum_{j=1}^p
\Pp\left(E_{gj}>\frac{u}{2}\right)
\nonumber\\
&\qquad\leq
2rp\exp\left\{-2n_{\min}\left(\frac{u}{2}\right)^2\right\}
=2rp\exp\left(-\frac{n_{\min}u^2}{2}\right),
\label{eq:proof-uniform-tail}
\end{align}
where the second inequality is the union bound and the last uses \eqref{eq:proof-dkw-one}. This proves \eqref{eq:uniform-tail}.

To obtain \eqref{eq:uniform-radius}, choose \(u=u_n(\eta)\). Then
\[
2rp\exp\left(-\frac{n_{\min}u_n(\eta)^2}{2}\right)
=2rp\exp\left\{-\log\left(\frac{2rp}{\eta}\right)\right\}
=\eta.
\]
Thus the complement of the event in \eqref{eq:uniform-radius} has probability at most \(\eta\). \(\square\)

\subsection{Proofs for support recovery and multiplicity}

\textbf{Proof of Theorem \ref{thm:screening-errors}.}
We first bound the probability of at least one false inclusion. If \(j\in\cN\), then \(\Delta_j=0\). By \eqref{eq:proof-delta-bound},
\begin{equation}
\label{eq:proof-null-delta}
  \wh\Delta_j
  =\abs{\wh\Delta_j-\Delta_j}
  \leq 2\max_gE_{gj}.
\end{equation}
Therefore the event \(\{\wh\Delta_j>\tau\}\) implies
\(
  \max_gE_{gj}>\tau/2.
\)
Using the union bound over clusters and then the Dvoretzky--Kiefer--Wolfowitz inequality,
\begin{align}
  \Pp(\wh\Delta_j>\tau)
  &\leq
  \sum_{g=1}^r\Pp(E_{gj}>\tau/2)
  \nonumber\\
  &\leq
  2r\exp\left\{-2n_{\min}(\tau/2)^2\right\}
  =2r\exp(-n_{\min}\tau^2/2).
\label{eq:proof-one-fp}
\end{align}
The event \(\wh\cS_\tau\not\subseteq\cS\) is the union of \(\{\wh\Delta_j>\tau\}\) over the \(p_0\) inactive coordinates. A second union bound gives
\[
  \Pp(\wh\cS_\tau\not\subseteq\cS)
  \leq 2rp_0\exp(-n_{\min}\tau^2/2),
\]
which proves \eqref{eq:false-inclusion}.

We next bound the probability of a false omission. Fix \(j\in\cS\). If \(j\notin\wh\cS_\tau\), then \(\wh\Delta_j\leq\tau\). Since \(\Delta_j\geq\Delta_{\min}>\tau\),
\begin{equation}
\label{eq:proof-miss-gap}
  \abs{\wh\Delta_j-\Delta_j}
  \geq \Delta_j-\wh\Delta_j
  \geq \Delta_{\min}-\tau.
\end{equation}
Combining \eqref{eq:proof-miss-gap} with \eqref{eq:proof-delta-bound}, a miss implies
\(
  \max_gE_{gj}\geq(\Delta_{\min}-\tau)/2.
\)
Hence
\begin{align}
\Pp(j\notin\wh\cS_\tau)
&\leq
\sum_{g=1}^r
\Pp\left(E_{gj}>\frac{\Delta_{\min}-\tau}{2}\right)
\nonumber\\
&\leq
2r\exp\left\{-2n_{\min}
\left(\frac{\Delta_{\min}-\tau}{2}\right)^2\right\}
\nonumber\\
&=2r\exp\left\{-\frac{n_{\min}(\Delta_{\min}-\tau)^2}{2}\right\}.
\label{eq:proof-one-fn}
\end{align}
Taking a union bound over the \(s\) active coordinates proves \eqref{eq:false-omission}.

The event \(\{\wh\cS_\tau\neq\cS\}\) is contained in the union of the false-inclusion and false-omission events. Adding the two bounds proves \eqref{eq:exact-recovery-bound}.

For the last statement, let \(u=u_n(\eta)\) and consider the event
\(
  \max_j|\wh\Delta_j-\Delta_j|\leq u,
\)
which has probability at least \(1-\eta\) by \Cref{thm:uniform}. If \(j\in\cN\), then \(\wh\Delta_j\leq u\), so it is not selected by the strict threshold \(\tau=u\). If \(j\in\cS\), then
\[
  \wh\Delta_j\geq\Delta_j-u
  \geq\Delta_{\min}-u>u=\tau
\]
whenever \(\Delta_{\min}>2u\). Thus every active coordinate and no inactive coordinate is selected on this event. \(\square\)

\textbf{Proof of Corollary \ref{cor:highdim}.}
Apply \Cref{thm:screening-errors} with \(\tau_n=Ca_n\). Since \(p_0\leq p_n\), the false-inclusion probability is bounded by
\begin{align}
2r_np_0\exp(-n_{\min,n}\tau_n^2/2)
&\leq
2r_np_n
\exp\left\{-\frac{C^2}{2}\log(r_np_n)\right\}
\nonumber\\
&=2(r_np_n)^{1-C^2/2}.
\label{eq:proof-cor-fp}
\end{align}
Because \(C>\sqrt2\), the exponent \(1-C^2/2\) is negative. Since \(r_np_n\to\infty\), the right-hand side tends to zero.

By \eqref{eq:beta-min-rate},
\[
  \Delta_{\min,n}-\tau_n
  \geq (C+D)a_n-Ca_n=Da_n.
\]
Since \(s\leq p_n\), the false-omission probability is at most
\begin{align}
2r_ns\exp\left\{-\frac{n_{\min,n}
(\Delta_{\min,n}-\tau_n)^2}{2}\right\}
&\leq
2r_np_n\exp\left\{-\frac{D^2}{2}\log(r_np_n)\right\}
\nonumber\\
&=2(r_np_n)^{1-D^2/2},
\label{eq:proof-cor-fn}
\end{align}
which tends to zero because \(D>\sqrt2\). The exact-recovery error is bounded by the sum of \eqref{eq:proof-cor-fp} and \eqref{eq:proof-cor-fn}, so it tends to zero.

If \(\Delta_{\min,n}\) is bounded below by a fixed positive constant and \(\log(r_np_n)=o(n_{\min,n})\), then \(a_n\to0\), and \eqref{eq:beta-min-rate} holds eventually for any fixed \(C,D\). This proves the final assertion. \(\square\)

\textbf{Proof of Proposition \ref{prop:fixed-alpha}.}
For each inactive coordinate, define
\(
  I_j=\1\{P_j\leq\alpha\}.
\)
Then
\(
  V_\alpha=\sum_{j\in\cN}I_j.
\)
By linearity of expectation, which does not require independence,
\begin{align}
  \E(V_\alpha)
  &=\sum_{j\in\cN}\E(I_j)
  =\sum_{j\in\cN}\Pp(P_j\leq\alpha)
  \leq\sum_{j\in\cN}\alpha
  =p_0\alpha.
\label{eq:proof-ev}
\end{align}
This proves \eqref{eq:ev}. If each null \(p\)-value is exactly uniform, then every inequality in \eqref{eq:proof-ev} is an equality.

Write
\[
  T_\alpha=\sum_{j\in\cS}\1\{P_j\leq\alpha\}
\]
to denote the number of active coordinates selected. Since
\(
  \wh m_\alpha=V_\alpha+T_\alpha,
\)
we have
\begin{align}
  \E(\wh m_\alpha)
  &=\E(V_\alpha)+\E(T_\alpha)
  \nonumber\\
  &\leq p_0\alpha+
  \sum_{j\in\cS}\Pp(P_j\leq\alpha)
  =p_0\alpha+
  \sum_{j\in\cS}\pi_j(\alpha),
\end{align}
which proves \eqref{eq:em}.

If the null \(p\)-values are exactly uniform and mutually independent, then the indicators \(I_j\), \(j\in\cN\), are independent Bernoulli random variables with common success probability \(\alpha\). Their sum is therefore \(\operatorname{Binomial}(p_0,\alpha)\), proving \eqref{eq:binomial}. Finally, Hoeffding's inequality for a sum of \(p_0\) independent random variables taking values in \([0,1]\) gives
\[
  \Pp\left(
  \abs{V_\alpha-\E V_\alpha}\geq t
  \right)
  \leq 2\exp(-2t^2/p_0).
\]
Under exact uniformity, \(\E V_\alpha=p_0\alpha\), yielding \eqref{eq:null-hoeffding}. \(\square\)

\textbf{Proof of Proposition \ref{prop:bonferroni}.}
A false inclusion occurs only if at least one inactive coordinate has \(P_j\leq q/p\). Therefore, by the union bound and super-uniformity,
\begin{align}
\Pp\left(\wh\cS_{\mathrm{Bonf}}(q)\not\subseteq\cS\right)
&=\Pp\left(\bigcup_{j\in\cN}\set{P_j\leq q/p}\right)
\nonumber\\
&\leq\sum_{j\in\cN}\Pp(P_j\leq q/p)
\nonumber\\
&\leq p_0\frac{q}{p}
\leq q.
\end{align}
No dependence assumption is used. \(\square\)

\subsection{Proofs for cluster preservation and label robustness}

\textbf{Proof of Theorem \ref{thm:cluster-integrity}.}
Fix a coordinate set \(T\supseteq\cS\) and a cluster label \(g\). Let
\(
  \eta_g(\bm X_{\cS})=\Pp(Z=g\mid\bm X_{\cS}).
\)
By \Cref{ass:sufficiency},
\begin{equation}
\label{eq:proof-suff-full}
  \Pp(Z=g\mid\bm X)
  =\eta_g(\bm X_{\cS})
  \quad\text{almost surely}.
\end{equation}
Because \(\bm X_T\) is measurable with respect to \(\bm X\), the tower property of conditional expectation gives
\begin{align}
  \Pp(Z=g\mid\bm X_T)
  &=\E\{\1(Z=g)\mid\bm X_T\}
  \nonumber\\
  &=\E\left[
     \E\{\1(Z=g)\mid\bm X\}
     \mid\bm X_T
  \right]
  \nonumber\\
  &=\E\left[
     \eta_g(\bm X_{\cS})
     \mid\bm X_T
  \right].
\label{eq:proof-tower}
\end{align}
Since \(T\supseteq\cS\), the random variable \(\eta_g(\bm X_{\cS})\) is measurable with respect to \(\bm X_T\). Hence the final conditional expectation in \eqref{eq:proof-tower} equals \(\eta_g(\bm X_{\cS})\). Combining this with \eqref{eq:proof-suff-full} proves \eqref{eq:posterior-preservation}.

For any action \(a\), the conditional risk based on a sigma-field generated by covariates is
\[
  \E\{L(a,Z)\mid\text{covariates}\}
  =\sum_{g=1}^rL(a,g)\Pp(Z=g\mid\text{covariates}).
\]
By \eqref{eq:posterior-preservation}, this conditional risk is the same function of \(a\) whether the covariates are \(\bm X_T\), \(\bm X_{\cS}\), or \(\bm X\). Therefore its infimum over \(a\) is the same almost surely in all three cases. Taking expectations proves \eqref{eq:risk-preservation}.

Because \(Z\) takes finitely many values, its conditional entropy is determined by the posterior vector:
\[
  H(Z\mid\text{covariates})
  =\E\left[-\sum_{g=1}^r
  \Pp(Z=g\mid\text{covariates})
  \log\Pp(Z=g\mid\text{covariates})\right],
\]
with the convention \(0\log0=0\). The posterior equality in \eqref{eq:posterior-preservation} therefore implies
\[
  H(Z\mid\bm X_T)=H(Z\mid\bm X_{\cS})=H(Z\mid\bm X).
\]
Since \(I(Z;W)=H(Z)-H(Z\mid W)\), subtracting these equal conditional entropies from the common finite entropy \(H(Z)\) proves \eqref{eq:mi-preservation}. \(\square\)

\textbf{Proof of Corollary \ref{cor:risk}.}
For any realized training sample, using more coordinates cannot increase the Bayes risk, because a decision rule based on a smaller set can always be regarded as a rule based on the larger set that ignores the additional coordinates. Hence
\begin{equation}
\label{eq:proof-bayes-monotone}
  R^*(\wh\cS_\tau)-R^*(\{1,\ldots,p\})\geq0.
\end{equation}
On the event \(\cS\subseteq\wh\cS_\tau\), \Cref{thm:cluster-integrity} gives
\(
  R^*(\wh\cS_\tau)=R^*(\{1,\ldots,p\}).
\)
On the complementary event, both Bayes risks lie in \([0,L_{\max}]\), so their nonnegative difference is at most \(L_{\max}\). Therefore, pointwise in the training sample,
\begin{equation}
\label{eq:proof-risk-indicator}
  0\leq
  R^*(\wh\cS_\tau)-R^*(\{1,\ldots,p\})
  \leq
  L_{\max}\1\{\cS\not\subseteq\wh\cS_\tau\}.
\end{equation}
Taking expectation over the training sample and applying \eqref{eq:false-omission} proves
\begin{align*}
  \E_{\mathrm{train}}[R^*(\wh\cS_\tau)-R^*(\{1,\ldots,p\})]
  &\leq L_{\max}\Pp(\cS\not\subseteq\wh\cS_\tau)\\
  &\leq 2L_{\max}rs
  \exp\left\{-\frac{n_{\min}(\Delta_{\min}-\tau)^2}{2}\right\}.
\end{align*}
This is \eqref{eq:expected-risk-bound}. \(\square\)

\textbf{Proof of Theorem \ref{thm:estimated-labels}.}
Define the true-label empirical separation
\(
  \wh\Delta_j=\max_{g<h}\norm{\wh F_{gj}-\wh F_{hj}}_\infty
\)
and the estimated-label separation
\(
  \wt\Delta_j=\max_{g<h}\norm{\wt F_{gj}-\wt F_{hj}}_\infty.
\)
For a fixed coordinate \(j\) and pair \(g<h\), the reverse triangle inequality gives
\begin{align}
&\left|
\norm{\wt F_{gj}-\wt F_{hj}}_\infty
-\norm{\wh F_{gj}-\wh F_{hj}}_\infty
\right|
\nonumber\\
&\qquad\leq
\norm{(\wt F_{gj}-\wh F_{gj})-(\wt F_{hj}-\wh F_{hj})}_\infty
\nonumber\\
&\qquad\leq
\norm{\wt F_{gj}-\wh F_{gj}}_\infty
+\norm{\wt F_{hj}-\wh F_{hj}}_\infty
\leq2\rho_n.
\label{eq:proof-label-pair}
\end{align}
Taking maxima over pairs and again using
\(
|\max a_k-\max b_k|\leq\max|a_k-b_k|
\)
yields
\begin{equation}
\label{eq:proof-label-delta}
  \max_j\abs{\wt\Delta_j-\wh\Delta_j}\leq2\rho_n.
\end{equation}
By \Cref{thm:uniform}, the event
\begin{equation}
\label{eq:proof-E-event}
  \mathcal E_1=
  \set{\max_j\abs{\wh\Delta_j-\Delta_j}\leq u_n(\eta)}
\end{equation}
has probability at least \(1-\eta\). By assumption, the event
\(
  \mathcal E_2=\{\rho_n\leq\bar\rho_n\}
\)
has probability at least \(1-\delta_n\). The union bound gives
\(
  \Pp(\mathcal E_1\cap\mathcal E_2)\geq1-\eta-\delta_n.
\)
On this intersection, for every \(j\),
\begin{align}
  \abs{\wt\Delta_j-\Delta_j}
  &\leq
  \abs{\wt\Delta_j-\wh\Delta_j}
  +\abs{\wh\Delta_j-\Delta_j}
  \nonumber\\
  &\leq2\bar\rho_n+u_n(\eta)
  =w_n(\eta),
\end{align}
proving \eqref{eq:estimated-label-uniform}.

If \(j\in\cN\), then \(\Delta_j=0\), so on the same event
\(
  \wt\Delta_j\leq w_n(\eta)\leq\tau
\)
whenever \(\tau\geq w_n(\eta)\). Such a coordinate is not selected by the strict inequality defining \(\wt\cS_\tau\), proving \(\wt\cS_\tau\subseteq\cS\).

If \(j\in\cS\), then
\[
  \wt\Delta_j
  \geq\Delta_j-w_n(\eta)
  \geq\Delta_{\min}-w_n(\eta).
\]
This quantity is larger than \(\tau\) whenever
\(
  \Delta_{\min}>\tau+w_n(\eta).
\)
Thus every active coordinate is selected. Setting \(\tau=w_n(\eta)\) makes the two conditions hold simultaneously when \(\Delta_{\min}>2w_n(\eta)\). \(\square\)

\textbf{Proof of Corollary \ref{cor:misassignment}.}
We first establish a deterministic bound for empirical distributions based on two index sets. Let \(A\) and \(B\) be nonempty finite sets of indices, with
\[
  C=A\cap B,
  \quad a=\abs{A\setminus B},
  \quad b=\abs{B\setminus A},
  \quad c=\abs C.
\]
Then \(N=\abs A=c+a\), \(M=\abs B=c+b\), and
\(
  d=\abs{A\triangle B}=a+b.
\)
For any function \(q_i\in[0,1]\),
\begin{align}
&\left|
\frac{1}{N}\sum_{i\in A}q_i
-\frac{1}{M}\sum_{i\in B}q_i
\right|
\nonumber\\
&\quad\leq
\left|\frac{1}{N}-\frac{1}{M}\right|
\sum_{i\in C}q_i
+\frac{1}{N}\sum_{i\in A\setminus B}q_i
+\frac{1}{M}\sum_{i\in B\setminus A}q_i
\nonumber\\
&\quad\leq
c\frac{|M-N|}{NM}+\frac{a}{N}+\frac{b}{M}.
\label{eq:proof-set-empirical}
\end{align}
Since \(|M-N|=|b-a|\leq a+b=d\) and
\(
  c/(NM)\leq1/\max(N,M)\leq1/\min(N,M),
\)
the first term in \eqref{eq:proof-set-empirical} is at most
\(
  d/\min(N,M).
\)
The sum of the last two terms is also at most
\(
  d/\min(N,M).
\)
Therefore
\begin{equation}
\label{eq:proof-set-tv}
  \left|
\frac{1}{N}\sum_{i\in A}q_i
-\frac{1}{M}\sum_{i\in B}q_i
\right|
\leq\frac{2d}{\min(N,M)}.
\end{equation}

Apply \eqref{eq:proof-set-tv} with \(A=I_g\), \(B=\wt I_g\), and
\(
  q_i=\1\{X_{ij}\leq x\}
\)
for arbitrary \(j\) and \(x\). Under \eqref{eq:set-misassignment},
\[
  d\leq\varepsilon_n n_g,
  \qquad
  \min(N,M)\geq(1-\varepsilon_n)n_g.
\]
Hence
\[
  \abs{\wt F_{gj}(x)-\wh F_{gj}(x)}
  \leq\frac{2\varepsilon_n}{1-\varepsilon_n}.
\]
The bound is uniform in \(x\), \(j\), and \(g\), proving \eqref{eq:rho-misassignment}. Substituting
\(
  \bar\rho_n=2\varepsilon_n/(1-\varepsilon_n)
\)
into
\(
  w_n(\eta)=u_n(\eta)+2\bar\rho_n
\)
gives \eqref{eq:w-misassignment}. \(\square\)

\subsection{Proofs for invariance, the PCA contrast, and bootstrap ranking}

\textbf{Proof of Proposition \ref{prop:monotone}.}
Fix coordinate \(j\) and write \(h=h_j\). Because \(h\) is a continuous strictly increasing bijection, it has a continuous strictly increasing inverse. For any cluster \(g\) and \(y\in\R\),
\begin{align}
  F_{gj}'(y)
  &=\Pp\{h(X_{gij})\leq y\}
  =\Pp\{X_{gij}\leq h^{-1}(y)\}
  =F_{gj}\{h^{-1}(y)\}.
\label{eq:proof-monotone-pop}
\end{align}
Thus, for every pair \(g<h\),
\begin{align}
  \norm{F_{gj}'-F_{hj}'}_\infty
  &=\sup_{y\in\R}
  \abs{F_{gj}\{h^{-1}(y)\}-F_{hj}\{h^{-1}(y)\}}
  \nonumber\\
  &=\sup_{x\in\R}\abs{F_{gj}(x)-F_{hj}(x)}
  =\norm{F_{gj}-F_{hj}}_\infty,
\label{eq:proof-monotone-pop-sup}
\end{align}
where the second equality uses that \(h^{-1}\) maps \(\R\) onto \(\R\). Taking the maximum over pairs proves \(\Delta_j'=\Delta_j\).

The empirical distribution transforms in exactly the same way:
\begin{align*}
  \wh F_{gj}'(y)
  &=\frac{1}{n_g}\sum_{i=1}^{n_g}
  \1\{h(X_{gij})\leq y\}\\
  &=\frac{1}{n_g}\sum_{i=1}^{n_g}
  \1\{X_{gij}\leq h^{-1}(y)\}
  =\wh F_{gj}\{h^{-1}(y)\}.
\end{align*}
Repeating \eqref{eq:proof-monotone-pop-sup} with empirical distribution functions gives \(\wh\Delta_j'=\wh\Delta_j\). Therefore every direct threshold comparison and the ordering of the statistics are unchanged. In every permutation replicate, the same pointwise transformation is applied after relabeling, so the permuted statistic is also unchanged. The permutation reference distribution and hence the permutation \(p\)-value are identical. \(\square\)

\textbf{Proof of Proposition \ref{prop:pca-counterexample}.}
Because \(\E Z=0\), \(\E\epsilon_j=0\), and all errors are independent of \(Z\), the mean of \(\bm X\) is zero. The variance of the first coordinate is
\begin{equation}
\label{eq:proof-var-x1}
  \operatorname{Var}(X_1)
  =\operatorname{Var}(\mu Z)+\operatorname{Var}(\epsilon_1)
  =\mu^2+\sigma_1^2.
\end{equation}
For \(j\geq2\), \(\operatorname{Var}(X_j)=\sigma_j^2\). Every off-diagonal covariance is zero. For example, for \(j\geq2\),
\[
  \operatorname{Cov}(X_1,X_j)
  =\E[(\mu Z+\epsilon_1)\epsilon_j]=0,
\]
and the noise coordinates are mutually independent. Hence the population covariance matrix is diagonal with entries
\(
  \mu^2+\sigma_1^2,\sigma_2^2,\ldots,\sigma_p^2.
\)

Under \eqref{eq:pca-variance-order}, at least \(q\) eigenvalues associated with noise coordinates exceed the eigenvalue associated with coordinate one. Therefore a leading \(q\)-dimensional principal eigenspace lies entirely in the span of the noise coordinates \(X_2,\ldots,X_p\). Any PCA score vector in this subspace is a function only of \((\epsilon_2,\ldots,\epsilon_p)\), which is independent of \(Z\). If \(U\) denotes the retained PCA representation, then
\[
  \Pp(Z=1\mid U)=\Pp(Z=1)=\frac12
\]
almost surely. Under zero-one loss, every classifier based on \(U\) has Bayes error
\(
  1-\max\{1/2,1/2\}=1/2.
\)

We now compute the marginal KS separation. Conditional on \(Z=1\),
\(
  X_1\sim N(\mu,\sigma_1^2),
\)
and conditional on \(Z=-1\),
\(
  X_1\sim N(-\mu,\sigma_1^2).
\)
Without loss of generality assume \(\mu>0\), since the separation depends only on \(|\mu|\). The difference between the two conditional distribution functions is
\begin{equation}
\label{eq:proof-normal-diff}
  d(x)
  =\Phi\left(\frac{x+\mu}{\sigma_1}\right)
   -\Phi\left(\frac{x-\mu}{\sigma_1}\right)>0.
\end{equation}
Its derivative is
\begin{equation}
\label{eq:proof-normal-derivative}
  d'(x)=\frac{1}{\sigma_1}
  \left[
  \varphi\left(\frac{x+\mu}{\sigma_1}\right)
  -\varphi\left(\frac{x-\mu}{\sigma_1}\right)
  \right],
\end{equation}
where \(\varphi\) is the standard normal density. If \(x<0\), then
\((x+\mu)^2<(x-\mu)^2\), so the first density in \eqref{eq:proof-normal-derivative} is larger and \(d'(x)>0\). If \(x>0\), the inequality is reversed and \(d'(x)<0\). Thus \(d\) is maximized at \(x=0\), where
\begin{align}
  d(0)
  &=\Phi(\mu/\sigma_1)-\Phi(-\mu/\sigma_1)
  =2\Phi(\mu/\sigma_1)-1.
\end{align}
Replacing \(\mu\) by \(|\mu|\) proves the first equality in \eqref{eq:normal-ks-signal}. For \(j\geq2\), the conditional distribution of \(X_j=\epsilon_j\) does not depend on \(Z\), so \(\Delta_j=0\). The active support is therefore \(\{1\}\). The final recovery statement follows by applying \Cref{thm:screening-errors} to this support. \(\square\)

\textbf{Proof of Theorem \ref{thm:bootstrap-ranking}.}
Conditional on the observed data, the variables \(W_{1j},\ldots,W_{Bj}\) are independent Bernoulli random variables with common mean \(\pi_j^*\). Hoeffding's inequality gives, for each fixed \(j\),
\begin{equation}
\label{eq:proof-bootstrap-one}
  \Pp^*\left(\abs{\wh\pi_j^*-\pi_j^*}>t\right)
  \leq2\exp(-2Bt^2).
\end{equation}
Taking a union bound over \(j=1,\ldots,p\) proves \eqref{eq:bootstrap-uniform}.

Now consider the event
\begin{equation}
\label{eq:proof-bootstrap-good}
  \mathcal B=\set{
  \max_j\abs{\wh\pi_j^*-\pi_j^*}<\gamma/2
  }.
\end{equation}
For any \(j\in A\) and \(k\notin A\), the gap assumption gives
\(
  \pi_j^*\geq\pi_k^*+\gamma.
\)
On \(\mathcal B\),
\begin{align*}
  \wh\pi_j^*
  &>\pi_j^*-\gamma/2
  \geq\pi_k^*+\gamma/2
  >\wh\pi_k^*.
\end{align*}
Thus every member of \(A\) is ranked above every nonmember. By \eqref{eq:bootstrap-uniform} with \(t=\gamma/2\),
\begin{align*}
  \Pp^*(\mathcal B^c)
  &\leq2p\exp\{-2B(\gamma/2)^2\}
  =2p\exp(-B\gamma^2/2).
\end{align*}
Therefore \(\Pp^*(\mathcal B)\) is at least the quantity in \eqref{eq:bootstrap-rank-prob}. \(\square\)

\subsection{Proofs for average dual information}

\textbf{Proof of Theorem \ref{thm:adi-properties}.}
By definition, \(\Gamma_f\in[0,1]\) and \(Q_f\in[0,1]\). Since \(\lambda\in[0,1]\), \(I_\lambda(f)\) is a convex combination of these two numbers. Therefore
\[
  0\leq\lambda\Gamma_f+(1-\lambda)Q_f\leq1,
\]
which proves part (i).

For part (ii), suppose first that \(\lambda\in(0,1)\) and \(I_\lambda(f)=1\). Then
\begin{equation}
\label{eq:proof-adi-one}
  0
  =1-I_\lambda(f)
  =\lambda(1-\Gamma_f)+(1-\lambda)(1-Q_f).
\end{equation}
Both terms on the right are nonnegative, and both coefficients are strictly positive. A sum of two nonnegative numbers is zero only if each number is zero. Hence \(\Gamma_f=1\) and \(Q_f=1\). Conversely, if both components equal one, their convex combination equals one. Similarly, if \(I_\lambda(f)=0\), then
\[
  0=\lambda\Gamma_f+(1-\lambda)Q_f
\]
is a sum of two nonnegative terms with positive coefficients, so \(\Gamma_f=Q_f=0\). The converse is immediate.

For part (iii), applying the same deterministic clustering rule to two identical copies of the original data gives the same partition. After the identity label alignment, the multiclass MCC is one, so \(\Gamma_{\mathrm{id}}=1\). Applying the same screening rule to the same coordinates gives the same selected support, and the dependency set of identity coordinate \(j\) is \(\{j\}\). The retained ancestry therefore contains every selected source coordinate, so \(Q_{\mathrm{id}}=1\). Part (iii) follows from the definition of \(I_\lambda\).

For part (iv), subtract the two ADI values:
\begin{align}
  I_\lambda(f)-I_\lambda(g)
  &=\lambda(\Gamma_f-\Gamma_g)
  +(1-\lambda)(Q_f-Q_g).
\label{eq:proof-adi-dominance}
\end{align}
Under the stated componentwise inequalities, both terms on the right are nonnegative, proving weak dominance. If \(\lambda\in(0,1)\) and at least one component inequality is strict, then the corresponding term in \eqref{eq:proof-adi-dominance} is strictly positive and the other is nonnegative, proving strict dominance.

For part (v), the formula
\[
  I_\lambda(f)=Q_f+\lambda(\Gamma_f-Q_f)
\]
shows directly that the function is affine in \(\lambda\). Moreover,
\begin{align*}
  \abs{I_\lambda(f)-I_\lambda(g)}
  &=\abs{
  \lambda(\Gamma_f-\Gamma_g)
  +(1-\lambda)(Q_f-Q_g)}\\
  &\leq
  \lambda\abs{\Gamma_f-\Gamma_g}
  +(1-\lambda)\abs{Q_f-Q_g},
\end{align*}
where the last step is the triangle inequality and the coefficients are nonnegative. This proves \eqref{eq:adi-lipschitz}. \(\square\)

\textbf{Proof of Proposition \ref{prop:adi-path}.}
For part (i), the assumptions give equality of both ADI components:
\(
  \Gamma_g=\Gamma_f
\)
and
\(
  Q_g=Q_f.
\)
Substituting these equalities into the definition yields
\[
  I_\lambda(g)
  =\lambda\Gamma_g+(1-\lambda)Q_g
  =\lambda\Gamma_f+(1-\lambda)Q_f
  =I_\lambda(f).
\]

For part (ii), let \(A_t\subseteq\cS\) be the active source coordinates covered by \(f_t\). Then
\(
  Q_t=|A_t|/s.
\)
If the transition to \(f_{t+1}\) covers exactly \(b_t\) previously uncovered active coordinates, nestedness implies
\(
  |A_{t+1}|=|A_t|+b_t.
\)
Therefore
\begin{equation}
\label{eq:proof-q-increment}
  Q_{t+1}-Q_t=\frac{b_t}{s}.
\end{equation}
Using the ADI definition and \eqref{eq:proof-q-increment},
\begin{align*}
  I_\lambda(f_{t+1})-I_\lambda(f_t)
  &=\lambda(\Gamma_{t+1}-\Gamma_t)
  +(1-\lambda)(Q_{t+1}-Q_t)\\
  &=\lambda(\Gamma_{t+1}-\Gamma_t)
  +(1-\lambda)\frac{b_t}{s},
\end{align*}
which is \eqref{eq:adi-increment}. If \(b_t=b\) and \(\Gamma_t\) is constant, every increment equals \((1-\lambda)b/s\), so the sequence is linear in \(t\). If \(\Gamma_{t+1}\geq\Gamma_t\), the first term is nonnegative, and the increment is at least \((1-\lambda)b_t/s\). \(\square\)

\textbf{Proof of Theorem \ref{thm:adi-consistency}.}
For notational simplicity, write
\[
  \wh S=\wh\cS_Y,
  \quad S=\cS,
  \quad \wh A=\wh\cA_f,
  \quad A=\cA_f,
  \quad \wh s=|\wh S|,
  \quad s=|S|.
\]
By definition,
\(
  a_n=|\wh S\triangle S|
\)
and
\(
  b_n=|\wh A\triangle A|.
\)
The cardinality difference between two finite sets is bounded by the cardinality of their symmetric difference, so
\begin{equation}
\label{eq:proof-cardinality-error}
  |\wh s-s|\leq a_n.
\end{equation}
The assumption \(a_n<s\) and \eqref{eq:proof-cardinality-error} imply
\begin{equation}
\label{eq:proof-shat-lower}
  \wh s\geq s-a_n>0.
\end{equation}
Thus \(\wh Q_f\) is well defined.

We next compare the intersection numerators. For arbitrary sets \(B_1,B_2,C_1,C_2\),
\begin{equation}
\label{eq:proof-intersection-symdiff}
  (B_1\cap C_1)\triangle(B_2\cap C_2)
  \subseteq
  (B_1\triangle B_2)\cup(C_1\triangle C_2).
\end{equation}
To verify the inclusion, take an element in the left-hand symmetric difference. It belongs to exactly one of the two intersections. If its membership in both first-coordinate sets and both second-coordinate sets were unchanged, it would belong to either both intersections or neither, a contradiction. Hence its membership differs in at least one pair, placing it in the right-hand union. Applying \eqref{eq:proof-intersection-symdiff} with \((B_1,C_1)=(\wh S,\wh A)\) and \((B_2,C_2)=(S,A)\) gives
\begin{equation}
\label{eq:proof-numerator-symdiff}
  |(\wh S\cap\wh A)\triangle(S\cap A)|
  \leq a_n+b_n.
\end{equation}
For finite sets \(U,V\),
\(
  ||U|-|V||\leq|U\triangle V|.
\)
Therefore, if
\(
  \wh N=|\wh S\cap\wh A|
\)
and
\(
  N=|S\cap A|,
\)
then
\begin{equation}
\label{eq:proof-numerator-error}
  |\wh N-N|\leq a_n+b_n.
\end{equation}

Now decompose the structural-coverage error:
\begin{align}
  |\wh Q_f-Q_f|
  &=\left|\frac{\wh N}{\wh s}-\frac{N}{s}\right|
  \nonumber\\
  &\leq
  \left|\frac{\wh N-N}{\wh s}\right|
  +N\left|\frac{1}{\wh s}-\frac{1}{s}\right|
  \nonumber\\
  &=\frac{|\wh N-N|}{\wh s}
  +N\frac{|s-\wh s|}{s\wh s}.
\label{eq:proof-q-decompose}
\end{align}
Since \(N\leq s\), equations \eqref{eq:proof-cardinality-error}, \eqref{eq:proof-shat-lower}, and \eqref{eq:proof-numerator-error} imply
\begin{align*}
  |\wh Q_f-Q_f|
  &\leq
  \frac{a_n+b_n}{s-a_n}
  +s\frac{a_n}{s(s-a_n)}\\
  &=\frac{2a_n+b_n}{s-a_n},
\end{align*}
which proves \eqref{eq:q-bound}.

Finally,
\begin{align*}
  |\wh I_\lambda(f)-I_\lambda(f)|
  &=\left|\lambda(\wh\Gamma_f-\Gamma_f)
  +(1-\lambda)(\wh Q_f-Q_f)\right|\\
  &\leq
  \lambda|\wh\Gamma_f-\Gamma_f|
  +(1-\lambda)|\wh Q_f-Q_f|\\
  &\leq
  \lambda|\wh\Gamma_f-\Gamma_f|
  +(1-\lambda)\frac{2a_n+b_n}{s-a_n},
\end{align*}
proving \eqref{eq:adi-bound}.

Under \eqref{eq:adi-consistency-conditions}, \(a_n/s\to0\) in probability, so with probability tending to one \(a_n<s/2\), say. On that event,
\[
  \frac{2a_n+b_n}{s-a_n}
  =\frac{2(a_n/s)+(b_n/s)}{1-a_n/s}
  \xrightarrow{\Pp}0.
\]
The agreement term also converges to zero by assumption. The upper bound in \eqref{eq:adi-bound} therefore converges to zero in probability, proving \(\wh I_\lambda(f)\to I_\lambda(f)\). \(\square\)


\begin{thebibliography}{99}

\bibitem[Bair et al.(2006)Bair, Hastie, Paul, and Tibshirani]{bair2006}
Bair, E., Hastie, T., Paul, D., and Tibshirani, R. (2006).
Prediction by supervised principal components.
\emph{Journal of the American Statistical Association} \textbf{101}, 119--137.

\bibitem[Benjamini and Hochberg(1995)]{benjaminihochberg1995}
Benjamini, Y. and Hochberg, Y. (1995).
Controlling the false discovery rate: A practical and powerful approach to multiple testing.
\emph{Journal of the Royal Statistical Society, Series B} \textbf{57}, 289--300.

\bibitem[B{\"o}hm and Hornik(2012)]{bohmhornik2012}
B{\"o}hm, W. and Hornik, K. (2012).
A Kolmogorov--Smirnov test for \(r\) samples.
\emph{Fundamenta Informaticae} \textbf{117}, 103--125.

\bibitem[Conover(1965)]{conover1965}
Conover, W. J. (1965).
Several \(k\)-sample Kolmogorov--Smirnov tests.
\emph{The Annals of Mathematical Statistics} \textbf{36}, 1019--1026.

\bibitem[Fan and Fan(2008)]{fanfan2008}
Fan, J. and Fan, Y. (2008).
High-dimensional classification using features annealed independence rules.
\emph{The Annals of Statistics} \textbf{36}, 2605--2637.

\bibitem[Fan and Lv(2008)]{fanlv2008}
Fan, J. and Lv, J. (2008).
Sure independence screening for ultrahigh dimensional feature space.
\emph{Journal of the Royal Statistical Society, Series B} \textbf{70}, 849--911.

\bibitem[Gorodkin(2004)]{gorodkin2004}
Gorodkin, J. (2004).
Comparing two \(K\)-category assignments by a \(K\)-category correlation coefficient.
\emph{Computational Biology and Chemistry} \textbf{28}, 367--374.

\bibitem[Jurman et al.(2012)Jurman, Riccadonna, and Furlanello]{jurman2012}
Jurman, G., Riccadonna, S., and Furlanello, C. (2012).
A comparison of MCC and CEN error measures in multi-class prediction.
\emph{PLoS ONE} \textbf{7}, e41882.

\bibitem[Kiefer(1959)]{kiefer1959}
Kiefer, J. (1959).
\(K\)-sample analogues of the Kolmogorov--Smirnov and Cram{\'e}r--von Mises tests.
\emph{The Annals of Mathematical Statistics} \textbf{30}, 420--447.

\bibitem[Li et al.(2012)Li, Zhong, and Zhu]{li2012}
Li, R., Zhong, W., and Zhu, L. (2012).
Feature screening via distance correlation learning.
\emph{Journal of the American Statistical Association} \textbf{107}, 1129--1139.

\bibitem[Mai and Zou(2013)]{maizou2013}
Mai, Q. and Zou, H. (2013).
The Kolmogorov filter for variable screening in high-dimensional binary classification.
\emph{Biometrika} \textbf{100}, 229--234.

\bibitem[Mai and Zou(2015)]{maizou2015}
Mai, Q. and Zou, H. (2015).
The fused Kolmogorov filter: A nonparametric model-free screening method.
\emph{The Annals of Statistics} \textbf{43}, 1471--1497.

\bibitem[Massart(1990)]{massart1990}
Massart, P. (1990).
The tight constant in the Dvoretzky--Kiefer--Wolfowitz inequality.
\emph{The Annals of Probability} \textbf{18}, 1269--1283.

\bibitem[Matthews(1975)]{matthews1975}
Matthews, B. W. (1975).
Comparison of the predicted and observed secondary structure of T4 phage lysozyme.
\emph{Biochimica et Biophysica Acta--Protein Structure} \textbf{405}, 442--451.

\bibitem[Witten and Tibshirani(2010)]{wittentibshirani2010}
Witten, D. M. and Tibshirani, R. (2010).
A framework for feature selection in clustering.
\emph{Journal of the American Statistical Association} \textbf{105}, 713--726.

\bibitem[Wolf and Naus(1973)]{wolfnaus1973}
Wolf, E. H. and Naus, J. I. (1973).
Tables of critical values for a \(k\)-sample Kolmogorov--Smirnov test statistic.
\emph{Journal of the American Statistical Association} \textbf{68}, 994--997.

\end{thebibliography}
\end{document}